\documentclass[sigconf]{acmart}
\renewcommand\footnotetextcopyrightpermission[1]{} 

\usepackage{amsmath, amsthm}
\usepackage{xcolor}
\usepackage{graphicx}
\usepackage{booktabs, makecell, tabularx}
\usepackage{rotating}
\usepackage[export]{adjustbox}
\usepackage{enumitem}
\usepackage{bbm}
\usepackage{listings}
\usepackage{cleveref}
\usepackage{subcaption}
\usepackage{mdframed}
\usepackage{comment}

\DeclareMathOperator*{\R}{\mathbb{R}}
\DeclareMathOperator*{\E}{\mathbb{E}}

\DeclareMathOperator{\mean}{mean}
\DeclareMathOperator{\sign}{sign}
\DeclareMathOperator{\Var}{Var}
\newcommand{\eqdef}{\mathbin{\stackrel{\rm def}{=}}}

\definecolor{teal}{rgb}{0.0,.5,.5}

\definecolor{indu}{rgb}{0.0,.5,.5}

\newtheorem{theorem}{Theorem}
\newtheorem{fact}{Fact}
\newtheorem{observation}{Observation}
\newtheorem{claim}{Claim}
\newtheorem{definition}{Definition}

\crefname{section}{Section}{Sections}
\crefname{equation}{Equation}{Equations}
\crefname{theorem}{Theorem}{Theorems}
\crefname{lemma}{Lemma}{Lemmas}
\crefname{listing}{Listing}{Listings}
\crefname{figure}{Figure}{Figures}

\newenvironment{myquote}[1]%
{\list{}{\leftmargin=#1\rightmargin=#1}\item[]}%
{\endlist}

\AtBeginDocument{%
	\providecommand\BibTeX{{%
			\normalfont B\kern-0.5em{\scshape i\kern-0.25em b}\kern-0.8em\TeX}}}

  \usepackage{nth}
  \usepackage{intcalc}

\title{How to Quantify Polarization in Models of Opinion Dynamics}

\author{Christopher Musco}
\email{cmusco@nyu.edu}
\affiliation{%
  \institution{New York University}
    \city{}\state{}\country{}
}

\author{Indu Ramesh}
\email{ir914@nyu.edu}
\affiliation{%
  \institution{New York University}
    \city{}\state{}\country{}
}
\author{Johan Ugander}
\email{jugander@stanford.edu}
\affiliation{%
  \institution{Stanford University}
   \city{}\state{}\country{}
}
\author{R. Teal Witter}
\email{rtealwitter@nyu.edu}
\affiliation{%
  \institution{New York University}
  \city{}\state{}\country{}
}

\begin{document}

\begin{abstract}
	It is widely believed that society is becoming increasingly polarized around important issues, a dynamic that does not align with common mathematical models of opinion formation in social networks. In particular, measures of polarization based on opinion variance are known to decrease over time in frameworks such as the popular DeGroot model. Complementing recent work that seeks to resolve this apparent inconsistency by modifying opinion models, we instead resolve the inconsistency by proposing changes to how polarization is quantified. 

    We present a natural class of group-based polarization measures that  capture the extent to which opinions are clustered into distinct groups. Using theoretical arguments and empirical evidence, we show that these group-based measures display interesting, non-monotonic dynamics, even in the simple DeGroot model. In particular, for many natural social networks, group-based metrics can increase over time, and thereby correctly capture perceptions of increasing polarization.

    Our results build on work by DeMarzo et al., who introduced a group-based polarization metric based on ideological alignment. We show that a central tool from that work, a limit analysis of individual opinions under the DeGroot model, can be extended to the dynamics of other group-based polarization measures, including established statistical measures
    like bimodality. 

    We also consider local measures of polarization that operationalize how polarization is perceived in a network setting. In conjunction with evidence from prior work that group-based measures better align with real-world perceptions of polarization, our work provides formal support for the use of these measures in place of variance-based polarization in future studies of opinion dynamics.

\end{abstract}



\maketitle
\pagestyle{plain} 

\section{Introduction}
Polarization of individual opinions and beliefs has become a topic of intense interest in recent years, especially in relation to politics \cite{democracies_divided_polarization_book}, and politically sensitive issues like climate change \cite{McCrightDunlap:2011} and public health \cite{Holone:2016,GreenEdgertonNaftel:2020}.
Polarization is often believed to threaten social stability; for example, it has been blamed for legislative deadlock \cite{Binder:2014,Bianco2020_polarization_congressional_deadlock}, decreased trust and engagement in the democratic process \cite{declining_trust_congress_polarization,LaymanCarseyHorowitz:2006,pernicious_polarization_harm_democracy}, and hindered responses to 
crises like the COVID-19 pandemic \cite{HartChinnSoroka:2020,GreenEdgertonNaftel:2020}.
In response to its impact, there is growing interest in using mathematical models of \emph{opinion dynamics} to formally study how polarization arises and evolves. Such models provide simple rules for how an individual's opinion on a topic changes in response to influence from that individual's social connections. Mathematical models of opinion dynamics offer a useful abstraction for studying important real-world phenomena \cite{AcemogluOzdaglar:2011}.
For example, they have been used to study the impact of biased assimilation \cite{HegselmannKrause2002,DandekarGoelLee:2013} and the effect of outside 
actors\footnote{
Actors like news agencies, social media companies, advertisers, and
governments can influence opinions in a social network by swaying
the strength of social connections, possibly by promoting or hiding
social media posts, creating fake user accounts and content, or running advertisements. 
By modeling these actions mathematically within an opinion dynamics framework, researchers can better understand how susceptible networks are to adversarial attacks \cite{AbebeChanKleinberg:2021,GeschkeLorenzHoltz:2019} and how ``filter bubbles'' emerge \cite{Pariser,chitra20analyzing}.}
on polarization  \cite{MuscoMuscoTsourakakis:2018,HazlaJinMossel:2019,GaitondeKleinbergTardos:2020,BrooksPorter:2020}.

To continue effectively leveraging such models, 
we first need to 
address a basic and important question:
\begin{myquote}{0.2in}\it
	How should the broad and imprecise  concept of polarization be {quantified} in mathematical models of opinion dynamics? 
\end{myquote}

Surprisingly, this question has received little attention. Most prior work defaults to quantifying polarization based on the overall \emph{variance} of societal opinions (opinions are typically encoded as real valued numbers) \cite{filterbubble,GaitondeKleinbergTardos:2020,BrooksPorter:2020}. While mathematically convenient, any variance-based approach faces a basic challenge: standard models of opinion formation in social networks, like the ubiquitous DeGroot learning model \cite{Degroot}, predict gradual convergence of opinion variance towards zero over time. This inevitable decrease stands in contradiction to the fact that, qualitatively, polarization is considered to exhibit far more interesting dynamics. For example, it is widely believed that polarization is currently \emph{increasing} across the globe on a variety of issues \cite{democracies_divided_polarization_book,McCrightDunlap:2011}, and that its dynamics have been impacted by forces such as the rise of social media \cite{Pariser}.

\subsection{Our Approach and Main Results}
Given the shortcomings of opinion variance as a measure of polarization, we address the central question of how to best quantify polarization by evaluating metrics through a \emph{dynamic lens}.
In particular, our goal is to identify natural metrics whose dynamics under simple models of opinion formation, like the DeGroot model, agree with {observed dynamics} of polarization in society. 

Towards this end, we introduce a class of \emph{group-based} metrics for polarization. We use ``group-based'' to reference the idea of a measure that is high when there are well-separated groups of individuals with different opinions. We formalize this notion in Section \ref{sec:prelims} by assuming shift- and scale-invariance, which are properties that naturally align with axiomatic treatments of clustering \cite{Kleinberg:2003}. For now, we leave ``group-based'' as an intuitive definition and illustrate with an example. Consider the following opinion vectors on a six node social network (each entry is one individual's opinion):
\vspace{-.5em}
\begin{align*}
	\mathbf{a} &= [-1, -.6, -.2, .2, .6,1 ] & 
	\mathbf{b} &= [.5,.5,.5,-.5,-.5,-.5]
\end{align*}
While $\mathbf{a}$ has larger variance than $\mathbf{b}$ (2.8 vs. 1.5), $\mathbf{b}$ would have higher polarization under a group-based measure, since opinions are more clearly clustered into two distinct groups. This cluster structure could be quantified, for example, by any statistical measure of bimodality,  like the popular Sarle's bimodality coefficient\footnote{See Definition \ref{def:bimodality} for more details and discussion.}. Sarle's coefficient is equal to $({\gamma^2 + 1})/{\kappa}$, where $\gamma$ is skewness and $\kappa$ is kurtosis, and evaluates to $0.58$ for $\mathbf{a}$, but a higher value of $1$ for $\mathbf{b}$.  

In this work, we explore two primary classes of related group-based measures:
\begin{description}[style=unboxed,leftmargin=0cm]
	\item[Statistical Measures (Section \ref{sec:stat_measures})] This class includes functions that, like Sarle's bimodality coefficient, measure group structure in an opinion distribution by looking at moments beyond the second (i.e., beyond variance). For example, the bimodality coefficient incorporates third and fourth moment information. 
	\item[Local Measures (Section \ref{sec:local_measures})] This class includes metrics that take into account local social connections on perceptions of group-structure.  For example, we study \emph{local agreement}, defined as the average percentage of an individuals social connections who agree on a particular topic (i.e., have an opinion on the same side of the mean). Networks with high local agreement may appear more polarized to individuals, who feel isolated in opinion bubbles. 
\end{description}
\smallskip

We show that these group-based measures of polarization behave very differently than  {variance-based measures}, exhibiting interesting, non-monotonic dynamics even in the simple DeGroot model. In particular, we prove that  the group-based measures converge to values that \emph{depend on the structure of the underlying social network}  governing the opinion dynamics. As a result, instead of always converging to zero like variance-based measures, they can \emph{increase} over time for certain networks.  
Our work builds on a result of DeMarzo, Vayanos, and Zwiebel \cite{DeMarzoVayanosZwiebel:2003}, who study another group-based metric, which we call ``ideological alignment''. Their work is based on an analysis of  the limiting behavior of a normalized vector containing each individuals divergence from the mean opinion under the DeGroot opinion dynamics model. We show that this analysis extends to both statistical and local group-based polarization measures. 

 

\begin{figure}
	\centering
	\includegraphics[width=\columnwidth]{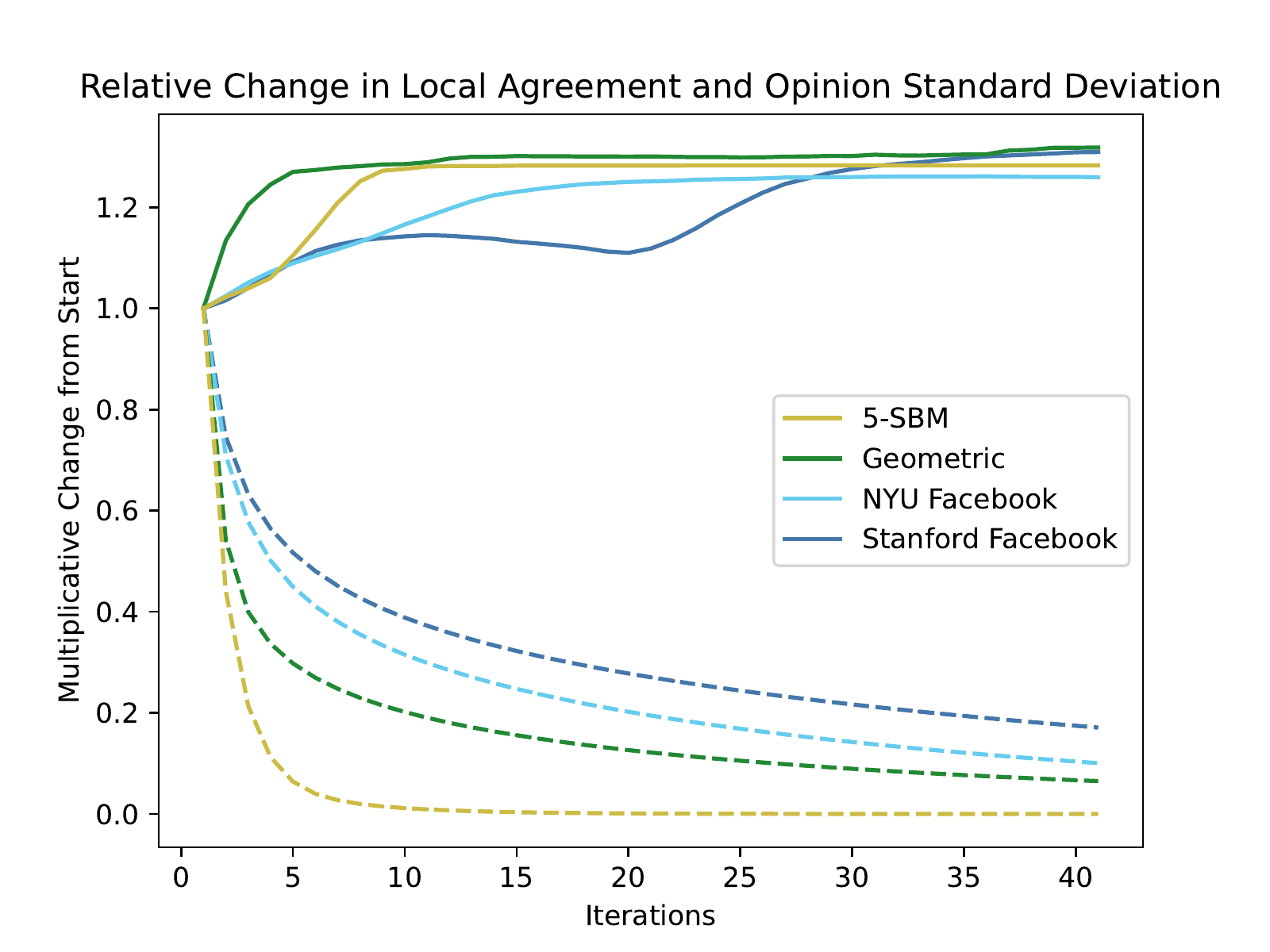}
	\caption[]{
		Number of iterations of DeGroot's model vs. two measures of polarization in synthetic and  real-world social networks.\footnotemark~ Dotted lines plot opinion standard deviation, a variance-based measure, while solid lines plot average local agreement, a natural group-based measure, discussed in Section \ref{sec:local_measures}. In contrast to standard deviation, for all networks the group-based measure increases over time, which is  consistent with real-world perceptions of how polarization can evolve. This finding provides evidence that the group-based measure may be a more appropriate method for quantifying polarization than the variance-based one.
	}
	\label{fig:stdvsagree}
\end{figure}
\footnotetext{NYU and Stanford are graphs from the
	Facebook100 data set \cite{traud2012social}. 
	5-SBM is a Stochastic Block Model graph on 1000 nodes split into
	five communities with the probability of intra- and inter-community edges equal to  $p=1/10$ and $q=1/100$, respectively.
	Geometric is a proximity graph with 1000 nodes
	on the unit square with $r=.1$ the cutoff for an edge, generated using the NetworkX Python library \cite{HagbergSchultSwart:2008}.}

Moreover, we demonstrate empirically that increases in group-based polarization are not only possible, but actually common in natural synthetic and real-world social networks. A sample result for average local agreement measure (discussed in Section \ref{sec:local_measures}) appears in Figure \ref{fig:stdvsagree}. 
Dotted lines plot opinion standard deviation, a variance-based measure of polarization. Solid lines plot average local agreement, a natural group-based measure of polarization, discussed in Section \ref{sec:local_measures}. In contrast to standard deviation, for all networks the group-based measure increases over time, which is  consistent with real-world perceptions of how polarization evolves.
We conclude that group-based measures not only have the capacity to model interesting dynamics, but also better align with perceptions of increasing polarization in reality. 

For specific group-based measures, we provide additional theoretical support for increasing polarization over time. For example, in Section \ref{sec:stat_measures}, we give a heuristic analysis for the limiting Sarle's bimodality of opinions in stochastic block-model graphs. We show that the equilibrium value of this measure under the DeGroot dynamics  is large for social network graphs with a small number of communities, a reasonable assumption of real-world networks. In Section \ref{sec:local_measures}, we also show that average local agreement in a social network converges to a value that depends on the second eigenvalue of the normalized adjacency matrix $\mathbf{D}^{-1}\mathbf{A}$. Polarization increases to a larger value when this eigenvalue is close to 1, which is empirically the case in a variety of real-world social network graphs. 





\subsection{Conclusions and Recommendations}
Our findings provide formal support for using group-based measures to quantify polarization in mathematical models of opinion dynamics.
The unrealistic monotonic dynamics of variance-based measures have led past studies to abandon simple opinion models like the DeGroot dynamics, and to adopt alternative, more complicated models to mathematically recover interesting polarization dynamics. For example, the Friedkin-Johnson dynamics \cite{fj}, bounded confidence model \cite{Lorenz:2007}, and geometric models have all seen recent attention \cite{HazlaJinMossel:2019,GaitondeKleinbergTardos:2021}.
A central conclusion of our work is that, alternatively, it may be the \emph{definition of polarization}, not the model, that lacks richness for understanding societal polarization. By turning from variance-based to natural group-based polarization measures, we see interesting dynamics even in the simplest models. 
 
Beyond our work, the recommendation to use group-based measures is also supported by empirical evidence that these measures are more aligned with how individuals perceive polarization in society  than variance-based metrics \cite{FiorinaAbrams:2008}. In particular, it has
been argued that perceived polarization does not correlate with significant \emph{absolute} differences in opinion (which drive overall opinion variance) \cite{LevenduskyMalhotra:2015,ChambersBaronInman:2006}. And while we cannot deny individual feelings of rising polarization in society\footnote{Polling data supports the widespread feeling that society is polarized. For example, $70\%$ of Americans believe that the United States is ``greatly divided'' when it comes to the country's most important values \cite{Monmouth-University-Polling-Institute:2019}.}, there is less evidence for increases in {absolute difference} in individual beliefs over time\footnote{A separate but related issue is the concept of \emph{affective polarization} \cite{Martin-Saveski:2022}, which asserts that individual responses to differences in opinions are also changing. In particular, feelings of hatred and distrust of those who hold opposite opinions seem to be increasing. There is evidence that group structure also contributes to these changes.} \cite{FiorinaAbramsPope:2005}. It follows that the concept of ``polarization'' may not equate to the absolute differences in opnion measured by opinion variance.
Instead, it has been argued that perceptions of polarization stem from perceptions of group-structure \cite{LevenduskyMalhotra:2015}. 
%
In fact, even the origin of the term ``polarization'' in the physical sciences suggests a group-based interpretation \cite{oed}. While sociological and psychological arguments for how to best quantify polarization are beyond the scope of this paper, the initial alignment between prior work and our findings is promising.

\subsection{Relation to Prior Work}
To the best of our knowledge, there has not been prior work on the problem of assessing the quality of various polarization metrics in mathematical models of opinion dynamics. Prior results have largely defaulted to variance-based measures. As mentioned, our work is most closely related to that of DeMarzo, Vayanos, and Zwiebel \cite{DeMarzoVayanosZwiebel:2003}, whose limit analysis we adopt directly (providing a new proof in Section \ref{sec:limit_analysis}). 
The main novelty of our work over \cite{DeMarzoVayanosZwiebel:2003} is two fold. First, we show that the approach from \cite{DeMarzoVayanosZwiebel:2003} has implications for a wide class of group-based polarization measures. In particular, we show that the limit result from that work implies that other group-based polarization measures converge to \emph{graph dependent} values, which \emph{can} be large for some graphs. Second, for different group-based measures, we provide experimental evidence and novel theoretical arguments that these group-based measures \emph{will} converge to large values for natural social networks. 

\section{Preliminaries}
\label{sec:prelims}


\smallskip\noindent\textbf{Graph Notation.} The DeGroot opinion dynamics model studied in this work is based on representing social connections via a weighted, undirected social graph, which we denote $G = (V,E)$. $G$ has  $|V| = n$ nodes and $|E| = m$ edges, possibly including self-loops. Let $\mathbf{A}$ be the adjacency matrix of $G$, with $A_{ij} = A_{ji} > 0$ if there is an edge between $i$ and $j$, and $A_{ij} = A_{ji} = 0$ otherwise. Let $\mathcal{N}(i)\subseteq \{1,\ldots, n\}$ denote the set of neighbors of node $i$, which includes all $j$ for which $A_{ij} \neq 0$. If $G$ contains a self-loop at node $i$ then $A_{ii} > 0$ and $i \in \mathcal{N}(i)$. Let $d_i = \sum_{j\in \mathcal{N}(i)} A_{ij}$ denote the degree of node $i$ and let $\mathbf{D}$ be a diagonal matrix containing $d_1, \ldots, d_n$ on its diagonal.

\smallskip\noindent\textbf{Vector Sign Normalization.} 
For a non-zero vector $\mathbf{x}$, let $[\mathbf{x}]^\pm$ denote the vector $\sign(x_i) \cdot \mathbf{x}$, where $x_i$ is the first non-zero entry in $\mathbf{x}$. That is, $[\mathbf{x}]^\pm$ is equal to either $+\mathbf{x}$ or $-\mathbf{x}$, with the sign chosen to ensure that the first non-zero entry is positive.

\subsection{DeGroot Opinion Dynamics}
Mathematical models of opinion dynamics have been studied for decades in economics \cite{hegselmann2005opinion, martins2008continuous,Jackson:2008}, applied math \cite{olfati2007consensus,jia2015opinion}, computer science \cite{urena2019review,DasGollapudiMunagala:2014}, and a variety of other fields \cite{goldman2002knowledge,noorazar2020recent}. We refer the reader to the survey in \cite{AcemogluOzdaglar:2011}. Such models typically view society as a graph, where nodes represent individuals and edges represent social connections of various strength. Connections can be in-person relationships, online relationships (e.g., edges in an online social network), or both. Simple rules and procedures then define how an individual's opinion on an issue (represented as a single discrete or continuous value, or as a vector) evolves over time. 

We focus on one of the earliest and most elegant models of opinion formation:  the DeGroot opinion dynamics \cite{French-Jr.:1956,Degroot}. This model is based on the idea that opinions on a topic, encoded as continuous values, propagate through the social graph via simple averaging. Nodes incorporate the beliefs of their neighbors into their own opinion over time. We formally describe the model below.

	\begin{definition}[DeGroot Opinion Dynamics] 
		\label{def:degroot}
		Let $G = (V,E)$ be a weighted, undirected graph\footnotemark~ with $n$ nodes, $m$ edges, adjaceny matrix $\mathbf{A}$, and degree matrix $\mathbf{D}$.
	For time steps $t = 0, 1, \ldots, T$, we associate the nodes of $G$ with an opinion vector $\mathbf{z}^{(t)} \in \R^n$ containing numerical values that represent each individual's current view on an issue. Starting with a fixed vector of initial opinions  $\mathbf{z}^{(0)} $, opinions under the DeGroot model evolve via the update:
	\begin{align*}
		z^{(t+1)}_{i} =  \frac{1}{D_{ii}}\sum_{j\in \mathcal{N}(i)}A_{ij}z^{(t)}_{i}, 
		\text{or equivalently, } &\mathbf{z}^{(t+1)} = \mathbf{D}^{-1}\mathbf{A}\mathbf{z}^{(t)}.
	\end{align*}
	\end{definition}

\footnotetext{The DeGroot model generalizes to directed graphs, but we consider the undirected case for simplicity.}

\medskip\noindent\textbf{Convergence to Consensus}. Like many other models of opinion dynamics, it is well known that the DeGroot dynamics converges to \emph{consensus} in the limit. Formally, we have:
\begin{fact}
	\label{fact:degroot_consensus}
	If $G$ is a connected, undirected, non-bipartite graph then, 
	\begin{align*}
		\mathbf{z}^* = \lim_{t\rightarrow \infty} \mathbf{z}^{(t)} = &c\cdot \vec{\mathbf{1}} & &\text{where} & c&= \sum_{i=1}^n \frac{d_i}{\sum_{j=1}^n d_j} z^{(0)}_i.
	\end{align*}
Note that $c$ is equal to the degree-weighted average opinion at time $0$.
\end{fact}

As discussed, a common approach to measuring  polarization on a single issue at time $t$ is to consider the overall opinion variance:
\begin{align*}
	\Var[\mathbf{z}^{(t)}] = \|\mathbf{z}^{(t)} - \mean(\mathbf{z}^{(t)})\cdot \vec{\mathbf{1}}\|_2^2
\end{align*}

Of course, if all opinions converge to the same constant value $c$, as guaranteed by Fact \ref{fact:degroot_consensus}, opinion variance eventually converges to zero. While this asymptotic observation only speaks to the model's behavior after a very long time, in the short term, variance also tends to decrease \emph{monotonically} with $t$, a fact that can be proven rigorously for regular graphs \cite{Dandekarpnas}.

\subsection{Group-based Polarization}
In this work, we study group-based polarization metrics, which we broadly define to include any function with three properties: invariant to a shift
in mean opinion, invariant to sign flips, and invariant to scaling. Formally:
\begin{definition}[Group-based Polarization]\label{def:groupbased}
	Let $f(G,\mathbf{z})$ be a function that maps an $n$-node graph $G$ and vector of opinions $\mathbf{z} \in \R^n$ to a measure of polarization. Then $f(G,\mathbf{z})$ is ``group-based'' if:
	\begin{enumerate}
		\item $f(G, \mathbf{z}) = f(-\mathbf{z})$,
		\item $f(G, \mathbf{z}) = f(\mathbf{z} + c\vec{\mathbf{1}})$ for any $\mathbf{z}$ and scalar $c$, and
		\item $f(G, \mathbf{z}) = f(c\mathbf{z})$ for any non-zero scalar $c$. 
	\end{enumerate}
\end{definition}
While variance and statistical measures studied in
Section \ref{sec:stat_measures} only depend on $\mathbf{z}$,
the local measures studied in Section \ref{sec:local_measures}
do depend on the underlying social network which is why we include
$G$ as a parameter to $f$.
Variance-based measures of polarization satisfy properties (1) and (2) of Definition \ref{def:groupbased}, but not (3). The last property reflects the fact that group structure 
should depend on \emph{relative} differences 
in opinions instead of absolute differences. 
For example, an opinion vector would be considered
polarized if we have two groups of individuals whose
mean opinions are further apart than opinions within each group, regardless
of absolute opinion difference.
The resulting requirement of scale-invariance has also appeared in axiomatic treatments of clustering objectives \cite{Kleinberg:2003}, which are closely related to group-based measures.





\section{Limit Analysis}\label{sec:limit_analysis}
While Fact \ref{fact:degroot_consensus}  implies that any variance-based measure of polarization converges to zero, this is not true for the group-based measures. Since they are both shift and scale invariant, they are insensitive to both the mean opinion (this is also true of variance-based measures) and to constant rescaling of the opinions.
As such, 
to analyze these measures, we prove a separate convergence result for the mean-centered, normalized opinion vector, which was also observed in \cite{DeMarzoVayanosZwiebel:2003}. In particular, we study the vector:
\begin{align*}
	\frac{\mathbf{z}^{(t)} - \mean(\mathbf{z}^{(t)}) \cdot \vec{\mathbf{1}} }{\|\mathbf{z}^{(t)} - \mean(\mathbf{z}^{(t)}) \cdot \vec{\mathbf{1}}\|_2}. 
\end{align*}
We show that, under mild conditions, in the DeGroot model this vector converges to a fixed function of the second eigenvector of the normalized social network adjacency matrix, $\mathbf{D}^{-1}\mathbf{A}$. We give a full proof below, which uses simpler arguments than \cite{DeMarzoVayanosZwiebel:2003}. 

\begin{theorem} \label{thm:second_eig_scaling}
Let $G$ be a connected graph with adjacency and degree matrices $\mathbf{A}$ and $\mathbf{D}$. Let $\mathbf{v}_1, \ldots, \mathbf{v}_n$ and $\lambda_1, \ldots, \lambda_n$ be the eigenvectors and eigenvalues of $\mathbf{D}^{-1}\mathbf{A}$, in order of magnitude. I.e.,  $|\lambda_1|\geq \ldots\geq |\lambda_n|$ . Let $\mathbf{z}^{(0)}, \ldots, \mathbf{z}^{(t)}$ be a sequence of opinion vectors updated via the DeGroot opinion dynamics as in Definition \ref{def:degroot}. Let $\mathbf{\bar{z}}^{(t)} = \mathbf{z}^{(t)} - \mean(\mathbf{z}^{(t)} )\cdot\vec{\mathbf{1}}$ be the mean-centered opinion vector at time $t$, and let $\mathbf{\bar{v}}_2 = \mathbf{v}_2 - \mean(\mathbf{v}_2)\cdot \vec{\mathbf{1}}$.
If $|\lambda_2| \neq |\lambda_3|$ and  $\langle \mathbf{D}^{1/2}\mathbf{v}_2, z^{(0)}\rangle \neq 0$ then:
\begin{align*}
	  \bar{\mathbf{s}}^* \eqdef \lim_{t\rightarrow \infty} \frac{[\mathbf{\bar{z}}^{(t)}]^{\pm}}{\|\mathbf{\bar{z}}^{(t)}\|_2} = \frac{[\mathbf{\bar{v}}_2]^{\pm}}{\|\mathbf{\bar{v}}_2\|_2}.
\end{align*}
Recall that for a non-zero vector $\mathbf{x}$, $[x]^{\pm}$ denotes $[x]^{\pm} = \sign(x_i)\cdot x$, where $x_i$ is the first non-zero entry in $\mathbf{x}$.
\end{theorem}
Theorem \ref{thm:second_eig_scaling} holds under two mild assumptions. First, we require that $\mathbf{z}^{(0)}$ has non-zero inner product with $\mathbf{D}^{1/2}\mathbf{v}_2$. This holds with probability 1 whenever $\mathbf{z}^{(0)}$ involves any non-zero isotropic random component. Second, we require that $|\lambda_2| \neq |\lambda_3|$, which will hold for any natural social network, as it can be guaranteed by assuming some randomness in the edges of the network\footnote{Informally, suppose we are given a fixed adversarial example network with $|\lambda_2| = |\lambda_3|$. A small random perturbation of the edges in the network will ensure that the second and third eigenvalue are no long \emph{exactly} equal with high probability.}. 

Under these conditions, Theorem \ref{thm:second_eig_scaling} shows that the normalized opinions converge to a vector $\bar{\mathbf{s}}^*$ that depends on the social graph $G$ (through its second eigenvector) but \emph{does not depend} on the initial opinion vector $\mathbf{z}^{(0)}$.

\begin{proof}
From the linear algebraic form of the DeGroot update rule, we have that:
\begin{align}
	\label{eq:power_method}
	\mathbf{z}^{(t)} = (\mathbf{D}^{-1}\mathbf{A})^t\mathbf{z}^{(0)} = \mathbf{D}^{-1/2} \left(\mathbf{D}^{-1/2}\mathbf{A}\mathbf{D}^{-1/2}\right)^t  \mathbf{D}^{1/2} \mathbf{z}^{(0)}.
\end{align}
Let $\mathbf{D}^{-1/2}\mathbf{A}\mathbf{D}^{-1/2} = \mathbf{V} \mathbf{\Sigma} \mathbf{V}^T$ denote the eigendecomposition of the symmetric normalized adjacency matrix $\mathbf{D}^{-1/2}\mathbf{A}\mathbf{D}^{-1/2}$. $\mathbf{\Sigma}$ is a diagonal matrix that contains real-valued eigenvalues identical to those of $\mathbf{D}^{-1}\mathbf{A}$. $\mathbf{V}$ is an orthogonal matrix whose columns contain eigenvectors $\mathbf{v}_1', \ldots, \mathbf{v}_n'$ where $\mathbf{v}_i' = \mathbf{D}^{1/2}\mathbf{v}_i/\|\mathbf{D}^{1/2}\mathbf{v}_i\|_2$. The eigenvalues of the normalized adjacency matrix of an undirected graph always lie in $[-1,1]$ and, since $\mathbf{A}$ is connected, there is exactly one eigenvector with eigenvalue $\lambda_1 = 1$.\footnote{There could be an eigenvalue with value $-1$, but we follow the convention that this would be denoted as $\lambda_2$.} It can be verified that the corresponding eigenvector of $\mathbf{D}^{-1/2}\mathbf{A}\mathbf{D}^{-1/2}$ is equal to $\mathbf{v}_1' = \mathbf{D}^{1/2}\vec{\mathbf{1}}/\| \mathbf{D}^{1/2}\vec{\mathbf{1}}\|_2$.

%
%
%
%
%


We expand \eqref{eq:power_method}, using that $(\mathbf{D}^{-1/2}\mathbf{A}\mathbf{D}^{-1/2})^t = \mathbf{V} \mathbf{\Sigma}^t \mathbf{V}^T$ since $\mathbf{V}$ is orthogonal. For $i = 1, \ldots, n$, let $c_i = \langle \mathbf{v}_i, \mathbf{D}^{1/2} \mathbf{z}^{(0)} \rangle$. We have that:
\begin{align*}
    \mathbf{z}^{(t)} = \mathbf{D}^{-1/2}\cdot \left(c_1\lambda_1^t\mathbf{v}_1 + c_2\lambda_2^t\mathbf{v}_2 + ...+c_n\lambda_n^t\mathbf{v}_n\right),
\end{align*}
and thus $\mathbf{\bar{z}}^{(t)} = \mathbf{{z}^{(t)}} - \mean\left(\mathbf{{z}}^{(t)}\right)$ equals:
 \begin{align*}
	\mathbf{\bar{z}}^{(t)} &= c_1\lambda_1^t\mathbf{D}^{-1/2}\mathbf{v}_1' - \mean\left(c_1\lambda_1^t\mathbf{D}^{-1/2}\mathbf{v}_1'\right)\cdot \vec{\mathbf{1}}\\ &+ c_2\lambda_2^t\mathbf{D}^{-1/2}\mathbf{v}_2' - \mean\left(c_2\lambda_2^t\mathbf{D}^{-1/2}\mathbf{v}_2'\right)\cdot \vec{\mathbf{1}} \\ &\hspace{7.2em}\vdots\\
	&+ c_n\lambda_n^t\mathbf{D}^{-1/2}\mathbf{v}_n' - \mean\left(c_n\lambda_n^t\mathbf{D}^{-1/2}\mathbf{v}_n'\right)\cdot \vec{\mathbf{1}}.
\end{align*}
Note that $\mathbf{D}^{-1/2}\mathbf{v}_1'$ is a scaling of the all ones vectors, so 
the first term in the sum above is zero. Letting $\mathbf{\bar{v}}_i = \mathbf{D}^{-1/2}\mathbf{v}_i' - \mean(\mathbf{D}^{-1/2}\mathbf{v}_i')\cdot \vec{\mathbf{1}}$, we are left with:
\begin{align*}
	\frac{\mathbf{\bar{z}}^{(t)}}{\|\mathbf{\bar{z}}^{(t)}\|_2} =  \frac{c_2\lambda_2^t\mathbf{\bar{v}}_2 +  c_3\lambda_3^t\mathbf{\bar{v}}_3 + \ldots + c_n\lambda_n^t\mathbf{\bar{v}}_n}{\|c_2\lambda_2^t\mathbf{\bar{v}}_2 + c_3\lambda_3^t\mathbf{\bar{v}}_3 + \ldots + c_n\lambda_n^t\mathbf{\bar{v}}_n\|_2}.
\end{align*}
We first note that $\|\mathbf{\bar{v}}_i\|_2 > 0$ for all $i = 2, \ldots, n$. To see why this is the case, observe that to have $\|\mathbf{\bar{v}}_i\|_2 = 0$, it must be that $\mathbf{v}_i' = c\mathbf{D}^{1/2}\vec{\mathbf{1}}$ for some constant $c$. However, this cannot be the case because $\mathbf{v}_i'$ is orthogonal to $\mathbf{v}_1' = c\mathbf{D}^{1/2}\vec{\mathbf{1}}$. Combined with our assumption that $c_2 = \langle \mathbf{D}^{1/2}\mathbf{v}_2, z^{(0)}\rangle \neq 0$, it follows that $\|c_2\mathbf{\bar{v}}_2\|_2 \geq 0$. Then, by our assumption that $|\lambda_2| \neq |\lambda_3$, we have $|\lambda_2| > |\lambda_i|$ for all $i=3, \ldots, n$. So, for any $\epsilon > 0$, there is some $t$ such that $\|c_3\lambda_3^t\mathbf{\bar{v}}_3 + \ldots + c_n\lambda_n^t\mathbf{\bar{v}}_n\|_2 \leq \epsilon \|c_2\mathbf{\bar{v}}_2\|_2$. We conclude that:
\begin{align*}
	\lim_{t\rightarrow \infty} \frac{[\mathbf{\bar{z}}^{(t)}]^{\pm}}{\|\mathbf{\bar{z}}^{(t)}\|_2}
	&=  \lim_{t\rightarrow \infty} \frac{[c_2\lambda_2^t\mathbf{\bar{v}}_2 + \ldots + c_n\lambda_n^t\mathbf{\bar{v}}_n]^{\pm}}{\|c_2\lambda_2^t\mathbf{\bar{v}}_2 + \ldots + c_n\lambda_n^t\mathbf{\bar{v}}_n\|_2} 
	=  \frac{[c_2\lambda_2^t\mathbf{\bar{v}}_2]^{\pm}}{\|c_2\lambda_2^t\mathbf{\bar{v}}_2\|_2}.
\end{align*}
This proves the theorem since $\frac{[c_2\lambda_2^t\mathbf{\bar{v}}_2]^{\pm}}{\|c_2\lambda_2^t\mathbf{\bar{v}}_2\|_2} = \frac{[\mathbf{\bar{v}}_2]^{\pm}}{\|\mathbf{\bar{v}}_2\|_2}$ for any $c_2, \lambda_2$.
\end{proof}

\subsection{Implications for Ideological Alignment}
In the work of DeMarzo, Vayanos, and Zwiebel \cite{DeMarzoVayanosZwiebel:2003},  Theorem \ref{thm:second_eig_scaling} is used to explain a phenomenon
involving \emph{multiple opinion vectors}, each 
defined for a different issue. They call the phenomenon ``unidimensional opinions'', but we prefer the terminology \emph{ideological
alignment}. 
Ideological alignment occurs when
large groups of individuals simultaneously 
differ in opinion on many issues
\cite{Evans:2003,AbramowitzSaunders:2008}. Also refereed to as ``party sorting'' \cite{FiorinaAbrams:2008}, this phenomenon is well-documented in the real-world, and there is strong survey-based evidence that it has increased in recent years \cite{levendusky2009partisan,Pew-Research-Center:2014}. 
Since it accentuates differences between groups, ideological alignment
has likely contributed to increased perception of polarization
\cite{AbramowitzSaunders:2008}.

{
	\setlength{\abovecaptionskip}{0pt}
\begin{figure}
	\centering
	\begin{subfigure}{.495\linewidth}
		\centering
		\caption{5-SBM}       
		\includegraphics[width=1\linewidth]{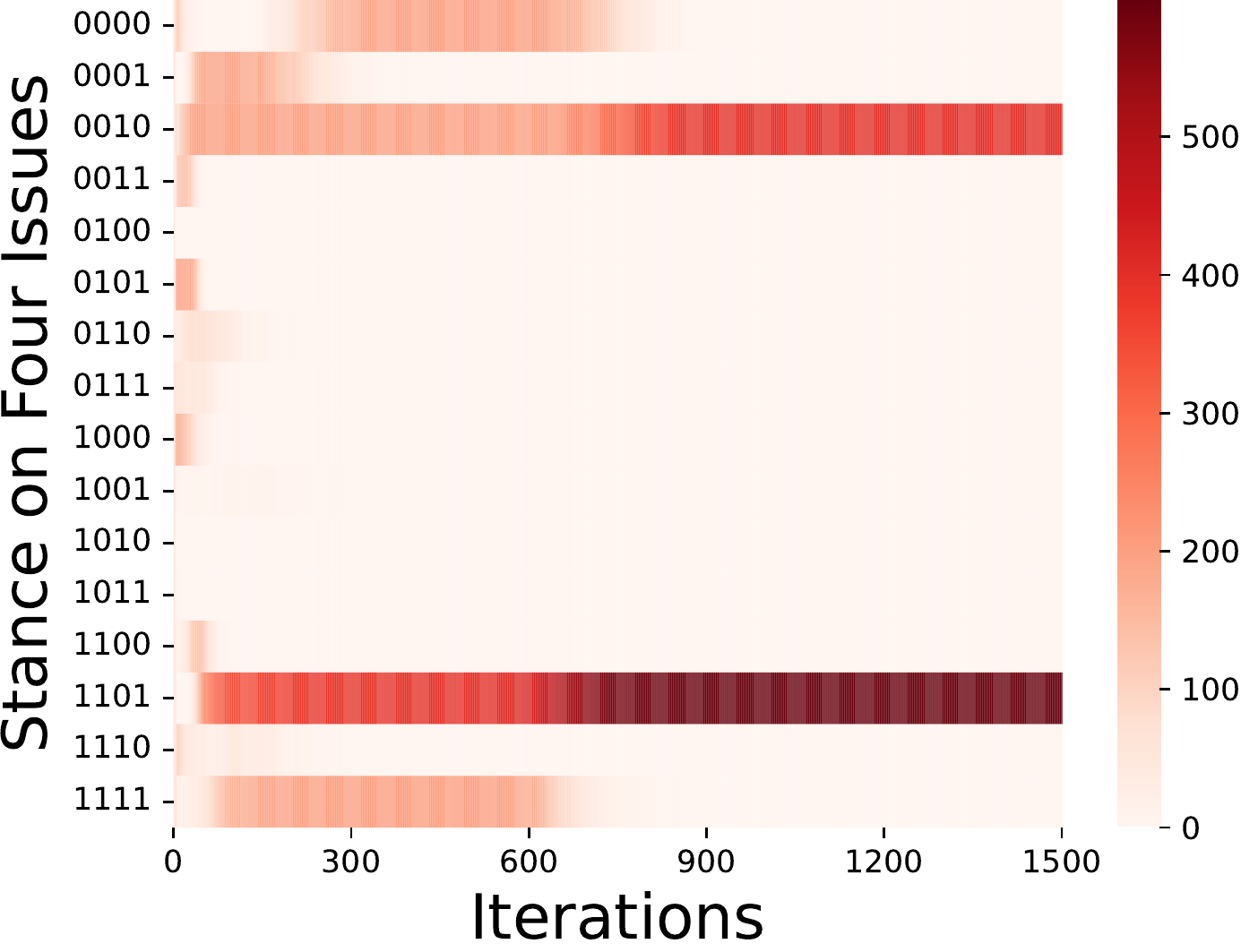}
		\label{fig:assortment-5sbm}
		\vspace{-1em}
	\end{subfigure}
	\begin{subfigure}{.495\linewidth}
		\centering
		\caption{Random Geometric}
		\includegraphics[width=1\linewidth]{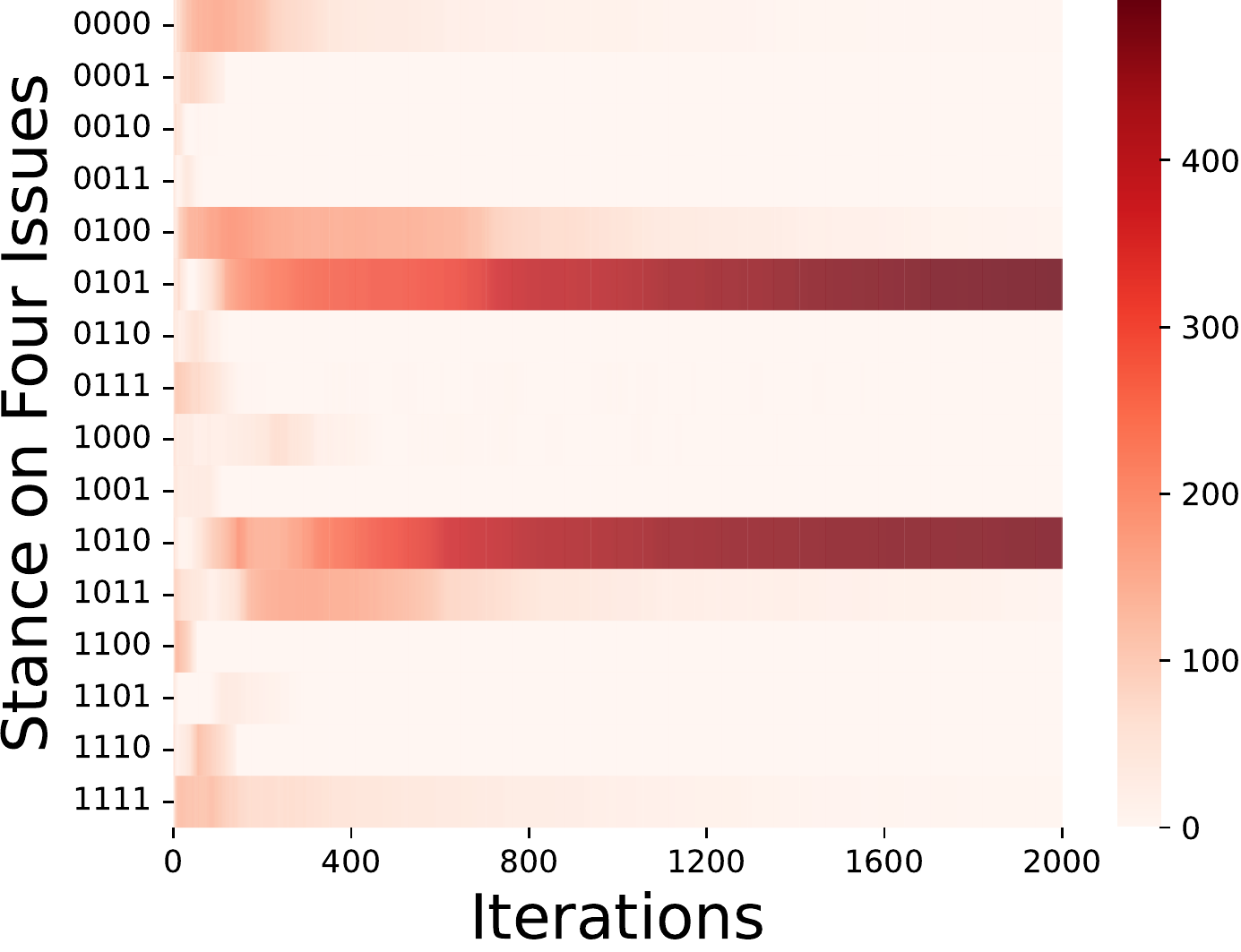}
		\label{fig:assortment-geometric}
		\vspace{-1em}
	\end{subfigure}          
	\begin{subfigure}{.495\linewidth}
		\caption{NYU Facebook}
		\centering
		\includegraphics[width=1\linewidth]{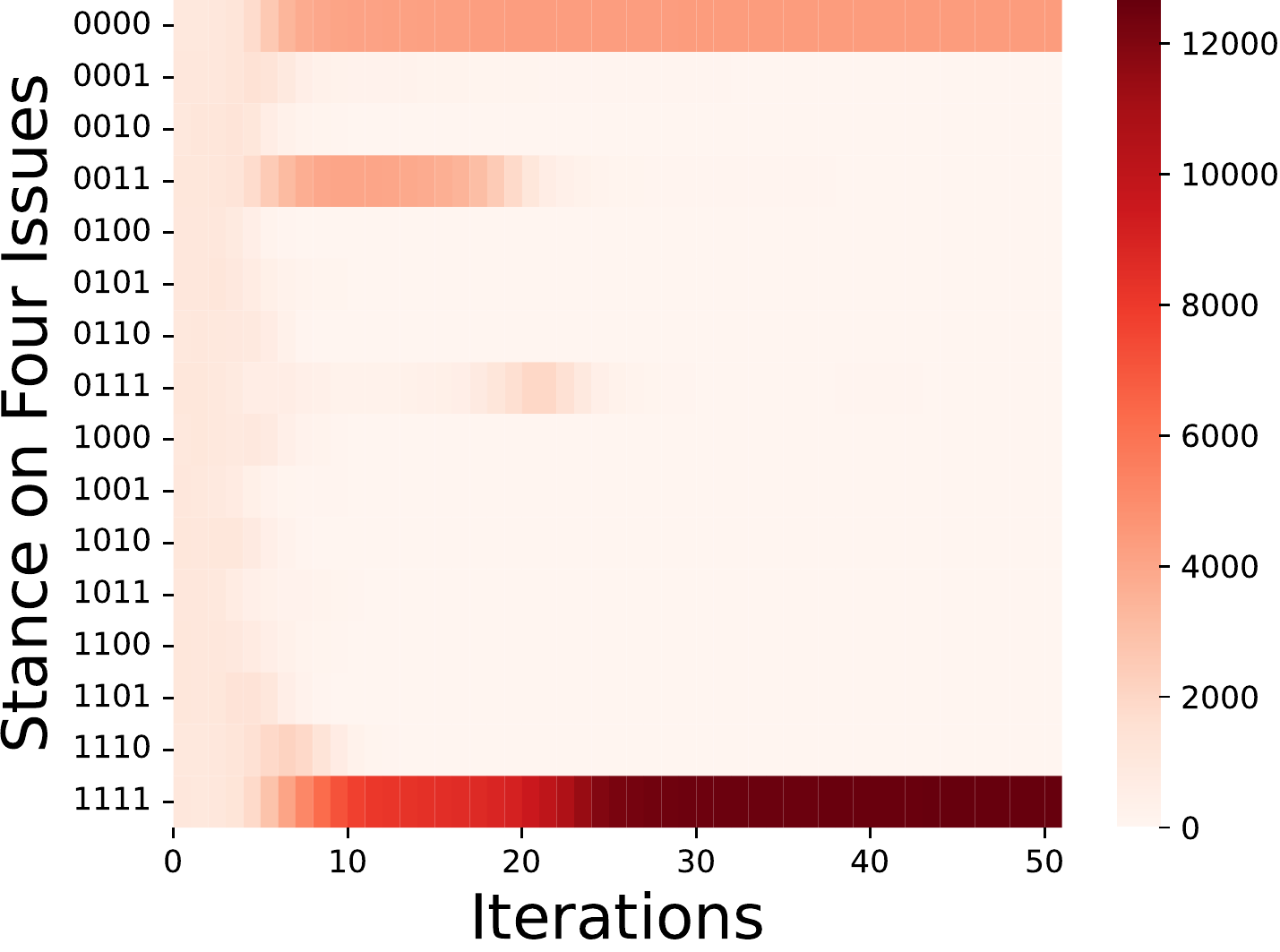}
		\label{fig:assortment-nyu}
	\end{subfigure}
	\begin{subfigure}{.495\linewidth}
		\centering
		\caption{Stanford Facebook}
		\includegraphics[width=1\linewidth]{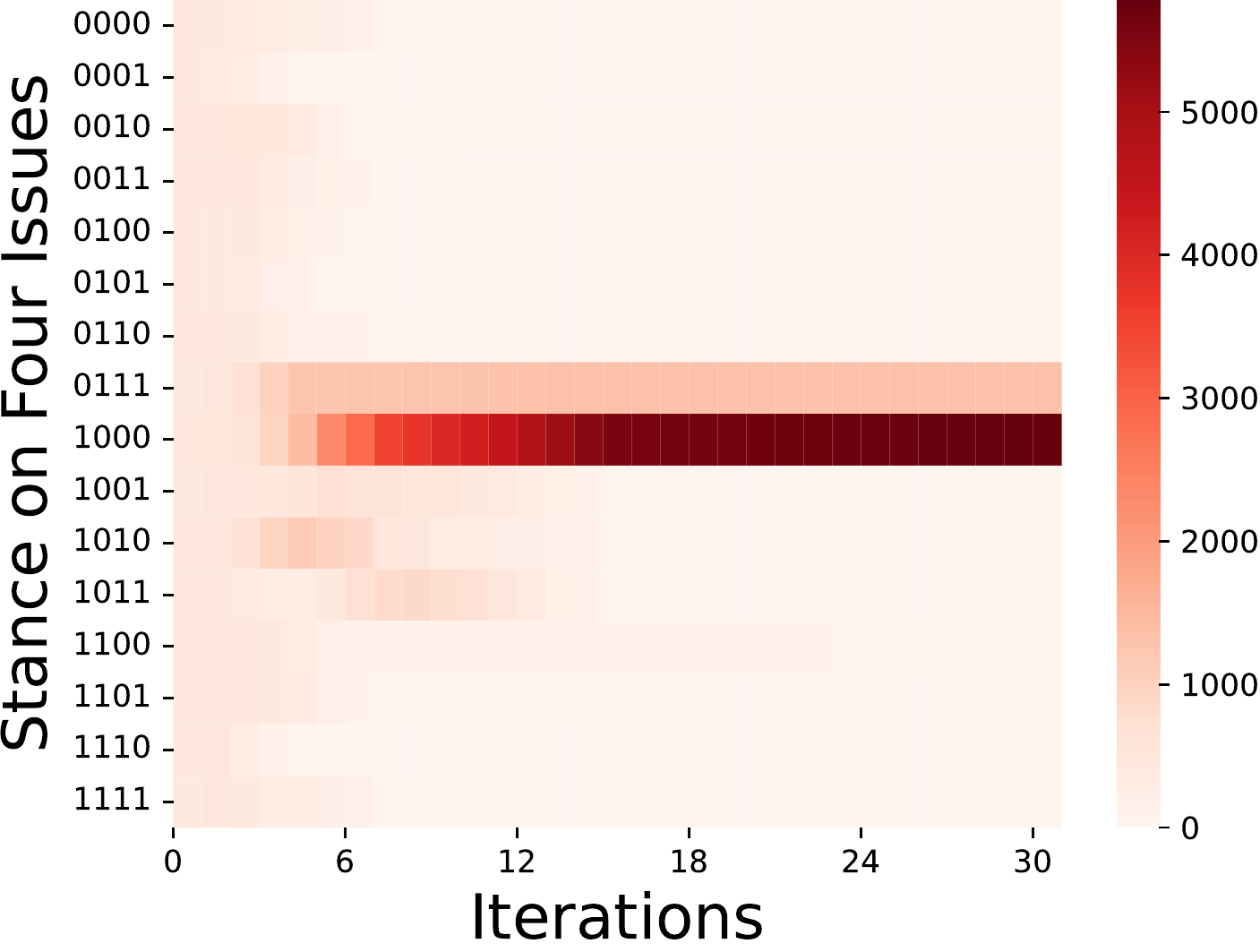}
		\label{fig:assortment-stanford}              
	\end{subfigure}   
	\caption{Heat map of ``binary opinion profiles'' across $m=4$ different issues
		by iteration of the DeGroot model. Each row corresponds to one of $2^m$ difference opinion profiles, and the color at time $t$ indicates the number of nodes whose current opinions match that profile. As predicted by Corollary \ref{corollary:alignment}, all individuals eventually sort into just two profiles, leading to perfect ideological alignment.
	}
	\label{fig:assortment}
\end{figure}
}
Theorem \ref{thm:second_eig_scaling} provides striking mathematical support for the emergence of ideological alignment. In particular, an an immediate corollary of the result is that, in the limit, individuals will perfectly sort into exactly two groups that simultaneously disagree on \emph{all issues} --i.e., for each issue, the members of one group will all have opinions on the opposite side of the mean as the other group.
We formalize their observation from \cite{DeMarzoVayanosZwiebel:2003} in
Corollary \ref{corollary:alignment}.

\begin{corollary}\label{corollary:alignment}
	Consider a social graph $G$ and $m$ different initial opinion vectors $z_1^{(0)}, \cdots, z_m^{(0)}$ satisfying the assumptions of Theorem \ref{thm:second_eig_scaling}.
	Apply the DeGroot opinion dynamics to each vector for
	$t$ steps to obtain opinions 
	$\mathbf{z}_1^{(t)}, \cdots, \mathbf{z}_m^{(t)}$, 
	and let 
	$\mathbf{s}_i^{(t)} = \sign(\mathbf{z}_i^{(t)} - 
	\mean(\mathbf{z}_i^{(t)} )\cdot \vec{\mathbf{1}})$. 
	Consider the matrix 
	$\mathbf{S}^{(t)} = [\mathbf{s}_1^{(t)}, \ldots, \mathbf{s}_m^{(t)}]$.
	In the limit as
	$t\rightarrow \infty$, $\mathbf{S}^{(t)}$
	will only contain {two unique rows}. 
\end{corollary}

Each row of $\mathbf{S}^{(t)}$ corresponds to a single node (individual) in $G$. It contains $\{+1,-1\}$ entries indicating if that individual has opinion below or above the mean for each of the $m$ topics at time $t$. The row can thus be viewed an individual's ``binary opinion profile''. The takeaway from Corollary \ref{corollary:alignment} is that, while there are $2^m$ possible opinion profiles, for large enough $t$ just two will dominate, becoming adopted by every individual. We visualized this alignment for four social networks in Figure \ref{fig:assortment}. Opinions were initialized randomly, so the rows of $\mathbf{S}^{(0)}$ are distributed evenly between all $2^m$ possible binary opinion profiles. However, as $t$ increases, we eventually see convergence to a state where  $\mathbf{S}^{(t)}$ has just two unique binary rows. The number of iterations until convergence varies by network.

While an interesting phenomenon, one potential limitation of ideological alignment as a polarization measure is that, like variance-based measures, it converges to the same extreme state for all social networks -- albeit to a state that is fully polarized instead of a fully in consensus.
In contrast, the other group-based measures of polarization discussed in this paper converge to \emph{network dependent quantities}, so their dynamics over time will natural differ within different social structure and can be impacted by outside influences that effect that structure, like social media or propaganda. 


\subsection{Implications for Group-Based Polarization}
The foundation of our work is the insight that
Theorem \ref{thm:second_eig_scaling} actually has implications on the limiting behavior of \emph{any} group-based
measure of polarization. Formally: 
\begin{corollary}\label{corollary:invariant}
	Let $f(G,\mathbf{z})$ be a group-based polarization metric according to Definition \ref{def:groupbased} that is continuous with respect to the argument $\mathbf{z}\in \R^n$. If the conditions of Theorem \ref{thm:second_eig_scaling} hold, 
	then 
	\begin{align*}
		\lim_{t \rightarrow \infty} f(G,\mathbf{z}^{(t)}) = f(G,\mathbf{v}_2) 
	\end{align*}
	where
	$z^{(t)}$ and $\mathbf{v}_2$ are as defined as in
	Theorem \ref{thm:second_eig_scaling}.
\end{corollary}
Corollary \ref{corollary:invariant} implies that, unlike variance-based measures which always converge to zero, under the mild assumptions of Theorem \ref{thm:second_eig_scaling}, any group-based measure of polarization converges to a value that depends on the social graph $G$.
 At the same time, the value does not depend on the starting opinions $\mathbf{z}^{(0)}$. With Corollary \ref{corollary:invariant} in place, we analyze several different group-based measures of polarization in the subsequent sections.

\section{Statistical Measures}
\label{sec:stat_measures}

 We start with statistical measures that, like variance, consider only the numerical values in an opinion vector $\mathbf{z}$, without taking into account the ordering of entries or their structure with respect to $G$. For example, the following common statistical measure of bimodality incorporates $3^\text{rd}$ and  $4^\text{th}$ moment information from $\mathbf{z}$:

\begin{definition}[Sarle's Bimodality Coefficient]
    \label{def:bimodality}
    Consider an opinion vector $\mathbf{z}$ and let
    $\mathbf{\bar{z}}$ denote $\mathbf{\bar{z}} = \mathbf{z} - \mean(\mathbf{z})$.
    Then the bimodality $\beta(\mathbf{z})$ is written in terms of the skewness $\gamma$
    and kurtosis $\kappa$ as follows:
    \begin{align}
        \beta(\mathbf{z}) = \frac{\gamma^2+1}{\kappa}
        \textrm{ where }
        \gamma = \frac{\mean(\mathbf{\bar{z}}^3)}{\mean(\mathbf{\bar{z}}^2)^{3/2}}
        \textrm{ and }
        \kappa = \frac{\mean(\mathbf{\bar{z}}^4)}{\mean(\mathbf{\bar{z}}^2)^2}.
        \nonumber
    \end{align}
\end{definition}

The bimodality coefficient of Definition \ref{def:bimodality} has been used as a measure of opinion polarization, e.g. in \cite{Paul-DiMaggioBryson:1996}, where it was compared against variance-based measures. The measure lies between $0$ and $1$, with $1$ indicating maximum polarization. However, even a random isotropic vector $\mathbf{r}$ (e.g., a vector with i.i.d. random Gaussian entries) will have bimodality $\beta(\mathbf{r})\approx 1/3$, since the skewness of a normal random variable is 0 and the kurtosis is 3. Accordingly, we consider a vector of opinions ``polarized'' if the bimodality is larger than $1/3$. 

\begin{figure}
	\centering
	\includegraphics[width=.9\columnwidth]{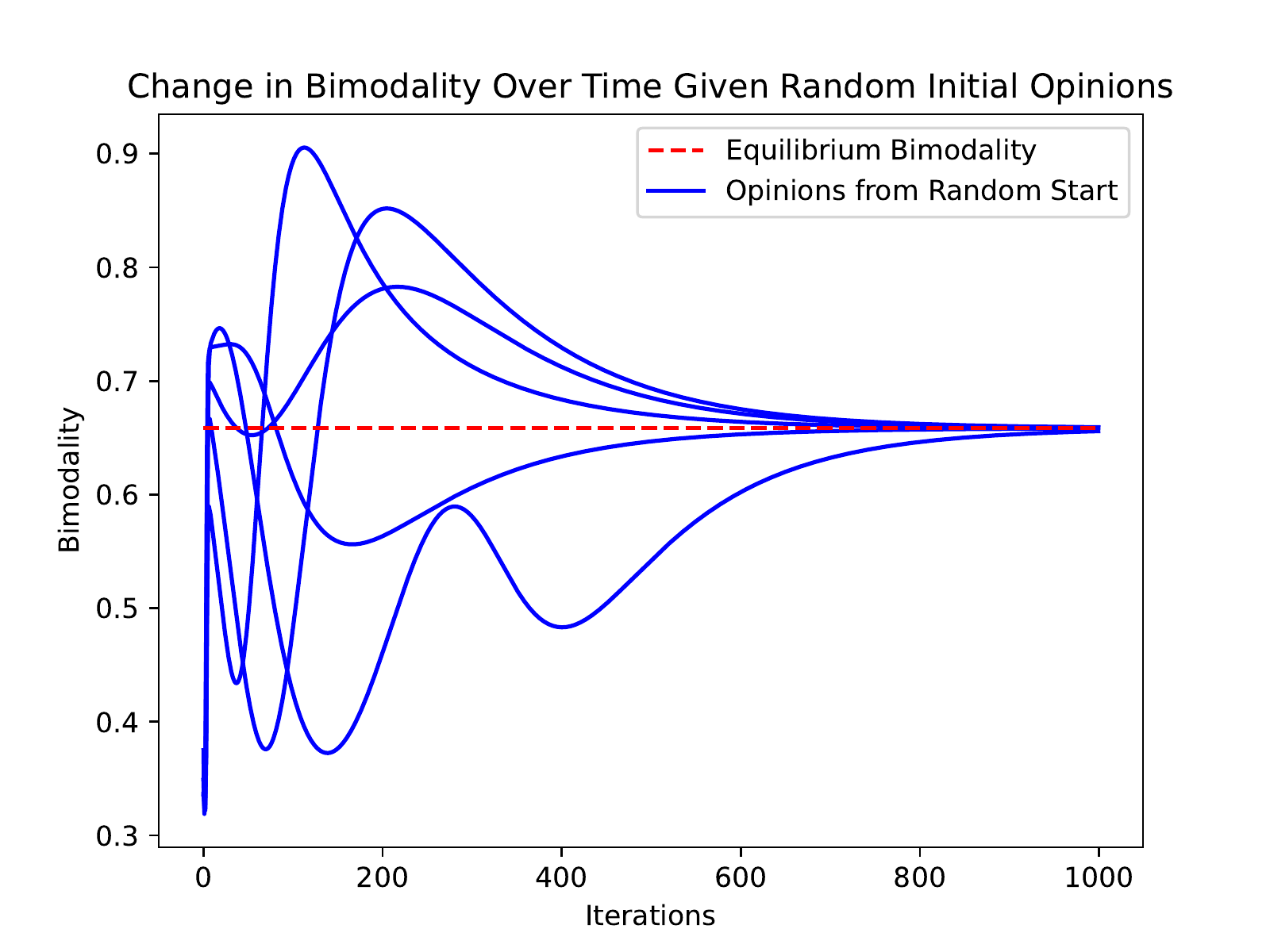}
	\caption{Opinion bimodality $\beta(\mathbf{z})$ plotted by iteration of DeGroot's
		opinion dynamics model run on the same 5 block SBM graph, and  initialized with five randomly generated starting
		opinion vectors. 
		As predicted by Corollary \ref{corollary:invariant}, in all cases opinion bimodality converges to a fixed non-zero
		equilibrium bimodality that depends on the graph.}
	\label{fig:bimodality-overtime}
\end{figure}

We demonstrate Corollary \ref{corollary:invariant} in Figure \ref{fig:bimodality-overtime}. We generate a Stochastic Block Model (SBM)
network \cite{HollandLaskeyLeinhardt:1983,Abbe:2017} on $n=1000$ nodes with five
communities (blocks). The probability
of an edge within a block is $p=1/10$ and the
probability of an edge between blocks is $q=1/100$.
We then initialize five random starting opinion vectors, each with i.i.d. standard normal entries.  We plot the bimodality of opinions as they evolve via the DeGroot dynamics.
By 1000 iterations, there is clear convergence to the bimodality
of the second eigenvector of the SBM, which, at $.658$, is much larger than the bimodalities of the starting opinions around $1/3$. So, while bimodality evolves in a highly non-monotonic way, it ultimately increases over time. 

\begin{table}[h!]
	\centering
	\begin{tabular}{|c c c|} 
		\hline
		$1^\text{st}$ Quartile & Median & $3^\text{rd}$ Quartile \\ 
		\hline
		.805 & .917 & .952 \\
		\hline
	\end{tabular}\vspace{.5em}
	\caption{Statistics of equilibrium bimodality for 100 college social networks from the Facebook100 data set \cite{traud2012social}.}
	
	\vspace{-1.5em}
	\label{table:facebook_bimod}
\end{table}

Increases in bimodality are even more pronounced in real-world social networks. We ran a similar experiment for 100 college social networks from the Facebook100 data set \cite{traud2012social} and observed that for all but five networks, bimodality \emph{increases} under the DeGroot dynamics with random starting opinions. The median and quartiles of the equilibrium bimodality (computed directly from the second eigenvector of each network) are included in Table \ref{table:facebook_bimod}. We conclude that the simple bimodality coefficient offers a clear contrast with variance-based measures of polarization that decrease over time. 

An informal analysis of SBM graphs offers theoretical support for increases in bimodality in natural social networks with a small number of well connected communities. Specifically, we argue that any SBM with a small number of blocks typically has equilibrium bimodality greater than $1/3$. We thus expect increasing bimodality under the DeGroot model if opinions are randomly initialized.

\begin{observation}\label{obs:sbm_bimodality}
	For a $k$-block SBM graph, the equilibrium bimodality is approximated by the {sample} bimodality of a normal random variable when $k$ samples are taken, which has expected value greater than $1/3$ for small $k$.
\end{observation}
We sketch a proof of Observation \ref{obs:sbm_bimodality}:
While the bimodality of the normal distribution is $1/3$, the empirical bimodality computed from a finite number of samples tends to over-estimate the true bimodality. While it is difficult to obtain an exact expression for the expectation of the sample bimodality, the sample kurtosis has expectation $3\frac{k-1}{k+1}$ \cite{JoanesGill:1998}. Sample kurtosis is thus an underestimate for small $k$, explaining the overestimate of bimodality, which depends on the inverse kurtosis. 
Now consider the expected symmetric normalized adjacency matrix $\mathbf{\bar{D}}^{-1/2}\mathbf{\bar{A}}\mathbf{\bar{D}}^{-1/2}$ of an SBM graph, where $\mathbf{\bar{D}} = \E[\mathbf{D}]$ and $\mathbf{\bar{A}} = \E[\mathbf{A}]$. 
It is not hard to see that the top $k$ eigenvectors of $\mathbf{\bar{D}}^{-1/2}\mathbf{\bar{A}}\mathbf{\bar{D}}^{-1/2}$ can be spanned by $\vec{\mathbf{1}}$ as well as $k$ block indicator vectors, each which is $1$ for the nodes in a single community, and $0$ for all other nodes. 
Since the actually normalized adjacency matrix $\mathbf{{D}}^{-1/2}\mathbf{{A}}\mathbf{{D}}^{-1/2}$ can be viewed as a perturbed version of  $\mathbf{\bar{D}}^{-1/2}\mathbf{\bar{A}}\mathbf{\bar{D}}^{-1/2}$, we roughly expect its first $k$ eigenvectors to also be spanned by $\vec{\mathbf{1}}$ and the $k$ block vectors -- a formal statement could be made by appealing to the Davis-Kahan perturbation theorem \cite{DavisKahan:1970}. Moreover, the $2^\text{nd}$ through the $(k-1)^\text{st}$ eigenvalues of $\mathbf{\bar{D}}^{-1/2}\mathbf{\bar{A}}\mathbf{\bar{D}}^{-1/2}$ are all the same, so we roughly expect the second eigenvector of $\mathbf{{D}}^{-1/2}\mathbf{{A}}\mathbf{{D}}^{-1/2}$ to be a \emph{random} linear combination of the $k$ block indicator vector, plus some scaling of $\vec{\mathbf{1}}$ (which has no impact on bimodality). If the random linear combination is isotropic, the second eigenvector will look exactly like $k$ samples from a random Gaussian distribution, each repeated $n/k$ times. This vector the same bimodality as $k$ random Gaussian samples.\footnote{Formally, we made an arguement about the second eigenvector of $\mathbf{{D}}^{-1/2}\mathbf{{A}}\mathbf{{D}}^{-1/2}$, whereas according to Corollary \ref{corollary:invariant}, it is the second eigenvector of $\mathbf{D}^{-1}\mathbf{A}$ that controls the equaibrium bimodality of a social network. However, since $\mathbf{{D}}$ is close to a scaling of the identity for an SBM graph, these two vectors will be very similar to each other.} 

Observation \ref{obs:sbm_bimodality} is visualized in Figure \ref{fig:bimodality-byk}, which was generated by computing the equilibrium opinion bimodality for 100 random $k$-SBM graphs with $1000$ and $2000$ nodes. While it approaches $1/3$ as $k$ increases, equilibrium bimodality is much larger for small $k$. We also plot the sample bimodality of $k$ i.i.d Gaussian samples (also computed using 100 trials), which as predicted by Observation \ref{obs:sbm_bimodality}, correlates well with the observed bimodality of the $k$-SBM.  

\begin{figure}
	\centering
	\includegraphics[width=.9\columnwidth]{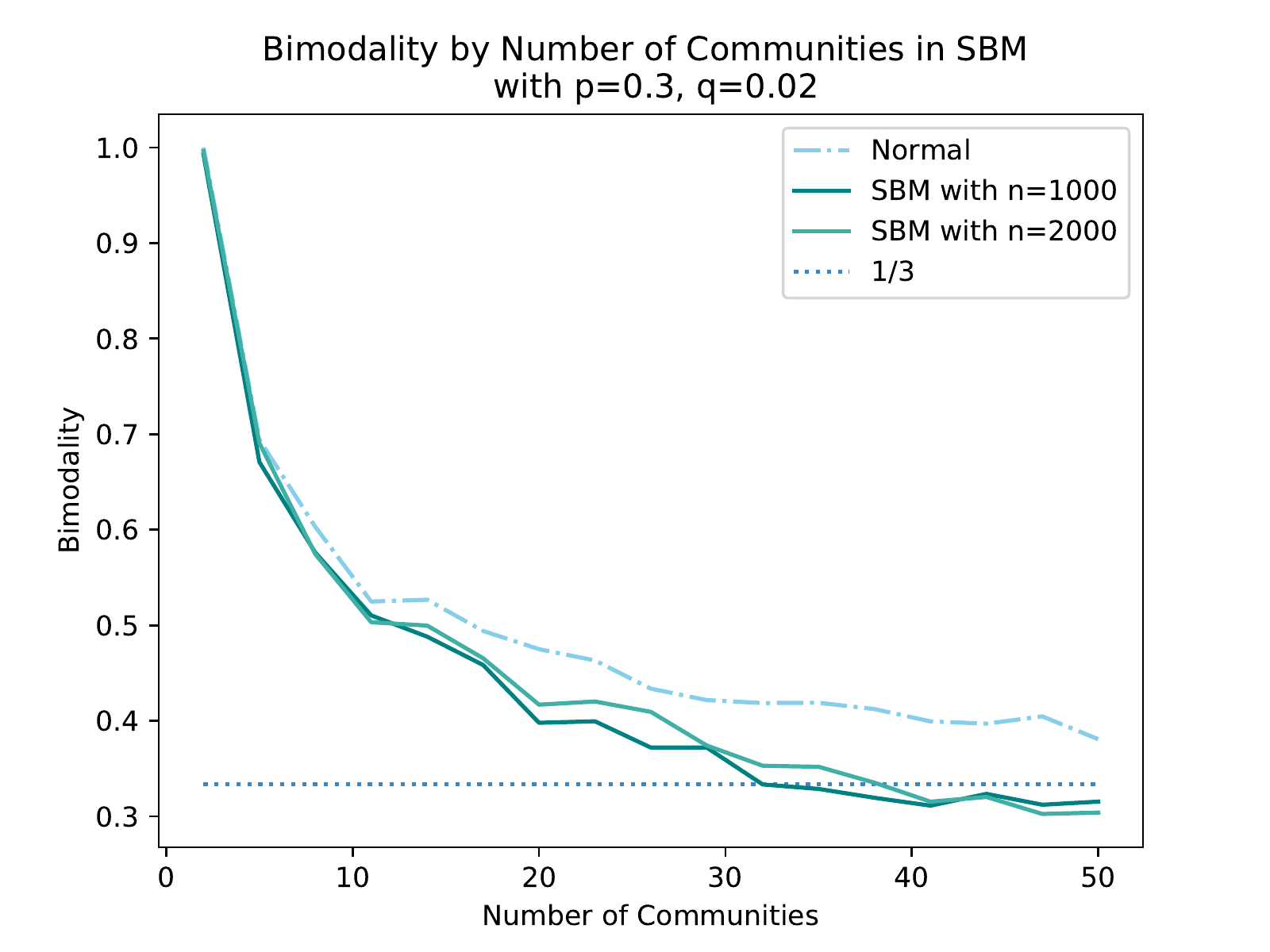}
	\vspace{-1em}
	\caption{Average equilibrium bimodality of $k$-SBM graphs with intra-block edge
		probability $3/10$ and inter-block edge probability $2/100$.
		Bimodality converges to the bimodality
		of a random normal variable for large $k$, which is $1/3$, but as predicted in Observation \ref{obs:sbm_bimodality}, can be much larger for small $k$.}
	\label{fig:bimodality-byk}
\end{figure}


\section{Local Measures}
\label{sec:local_measures}
Another interesting class of group-based polarization measures are those that take into account local structure of the social graph $G$. Such measures are motivated by the fact that individuals are most heavily exposed to the opinions of their social connections -- i.e., their neighborhood in $G$. Individuals likely also have a sense of the overall mean opinion in $G$ (e.g., from the news), but do not simultaneously sense all opinions in a social network. 

In this section we introduce and study one such measure, which we call \emph{average local agreement} that takes these considerations into account. In particular, we define the local agreement of a vertex $i$ to be the ratio of $i$'s neighbors whose opinion falls on the same side (above or below) the mean opinion $\mean(\mathbf{z})$ as $i$.
We posit that \emph{high local agreement} correlates with \emph{high perceptions of polarization}, as individuals who feel more isolated in a group, away from those differing opinion, tends to experience feelings of polarization \cite{LevenduskyMalhotra:2015}. 

We formally define average local agreement below. We use $\sign(\mathbf{x})$ to denote the operation that rounds every entry of a vector $\mathbf{x}$ to +1 or -1, taking the convention that if $x_i = 0$, $[\sign(\mathbf{x})]_i = +1$. 
\begin{definition}[Average Local Agreement] \label{def:local_agreement}
	Let $G$ be a social network on $n$ nodes and let $\mathbf{z}\in \R^n$ be an opinion vector.  Let $\mathbf{s}= \sign(\mathbf{z} - \mean(\mathbf{z}) \cdot \vec{\mathbf{1}})$. The average local agreement $\mathcal{L}(G,\mathbf{z})$ equals:
	\begin{align*}
		\mathcal{L}(G,\mathbf{z}) &= \frac{1}{n}\sum_{i=1}^n \frac{1}{d_i}\sum_{j\in\mathcal{N}(i)} \mathbbm{1}[s_i = s_j] & \\ &\text{where \,\,} \mathbf{s} = \sign(\mathbf{z} - \mean(\mathbf{z}) \cdot \vec{\mathbf{1}}).
	\end{align*}
Recall that $\mathcal{N}(i)$ denotes the neighborhood of node $i$ in $G$, and $\mathbbm{1}[\cdot]$ is an indicator function that evaluates to $1$ if the expression in brackets is true, and to $0$ otherwise. 
\end{definition}

Like bimodality, average local agreement is a group-based measure as specified in Definition \ref{def:groupbased}. So, as in Corollary \ref{corollary:invariant}, we have that in the DeGroot dynamics, under the assumptions of Theorem \ref{thm:second_eig_scaling}, $\lim_{t\rightarrow \infty} \mathcal{L}(G,\mathbf{z}^{(t)}) =  \mathcal{L}(G,\mathbf{v}_2)$,
where $z^{(t)}$ and $\mathbf{v}_2$ are as defined as in
the theorem.
Average local agreement is bounded between $[0,1]$ and we expect a value of $1/2$ for randomly initialized opinions. So, any value above $1/2$ is considered ``polarized''. 
As shown in Table \ref{table:facebook_local}, we observe very high average local agreement in the limit for real-world social networks. For all but two of the 100 networks in the Facebook100 data set, this measure of polarization converged to a value above $.6$, and was typically well above $.9$. In Figure \ref{fig:local} we also visualize local agreement over time for a random 5-SBM graph and a random geometric graph, as well as the Swarthmore Facebook graph (chosen for its small size). In all cases, ``bubbles'' of high local agreement visibly emerge, with average local agreement increasing to $.785$, $.954$, and $.941$ for the three graphs, respectively. 

\begin{table}[h!]
	\centering
	\begin{tabular}{|c c c|} 
		\hline
		$1^\text{st}$ Quartile & Median & $3^\text{rd}$ Quartile \\ 
		\hline
		.904 & .947 & .960 \\
		\hline
	\end{tabular}\vspace{.5em}
	\caption{Statistics of the equilibrium average local agreement, $\mathcal{L}(G,\mathbf{v}_2)$ for the Facebook100 data set \cite{traud2012social}.}
	
	\vspace{-1.5em}
	\label{table:facebook_local}
\end{table}

\begin{figure}
	\setlength\tabcolsep{0pt}
	\settowidth\rotheadsize{\footnotesize Swarthmore}
	\begin{tabularx}{\linewidth}{l XXX|X }
		& \thead{0 Iterations} & \thead{5 Iterations} & \thead{Equilibrium} & \thead{Equilibrium\\ Opinions}  \\
		\rothead{\centering\footnotesize 5 Block SBM}        
		&   \includegraphics[width=\hsize,valign=m]{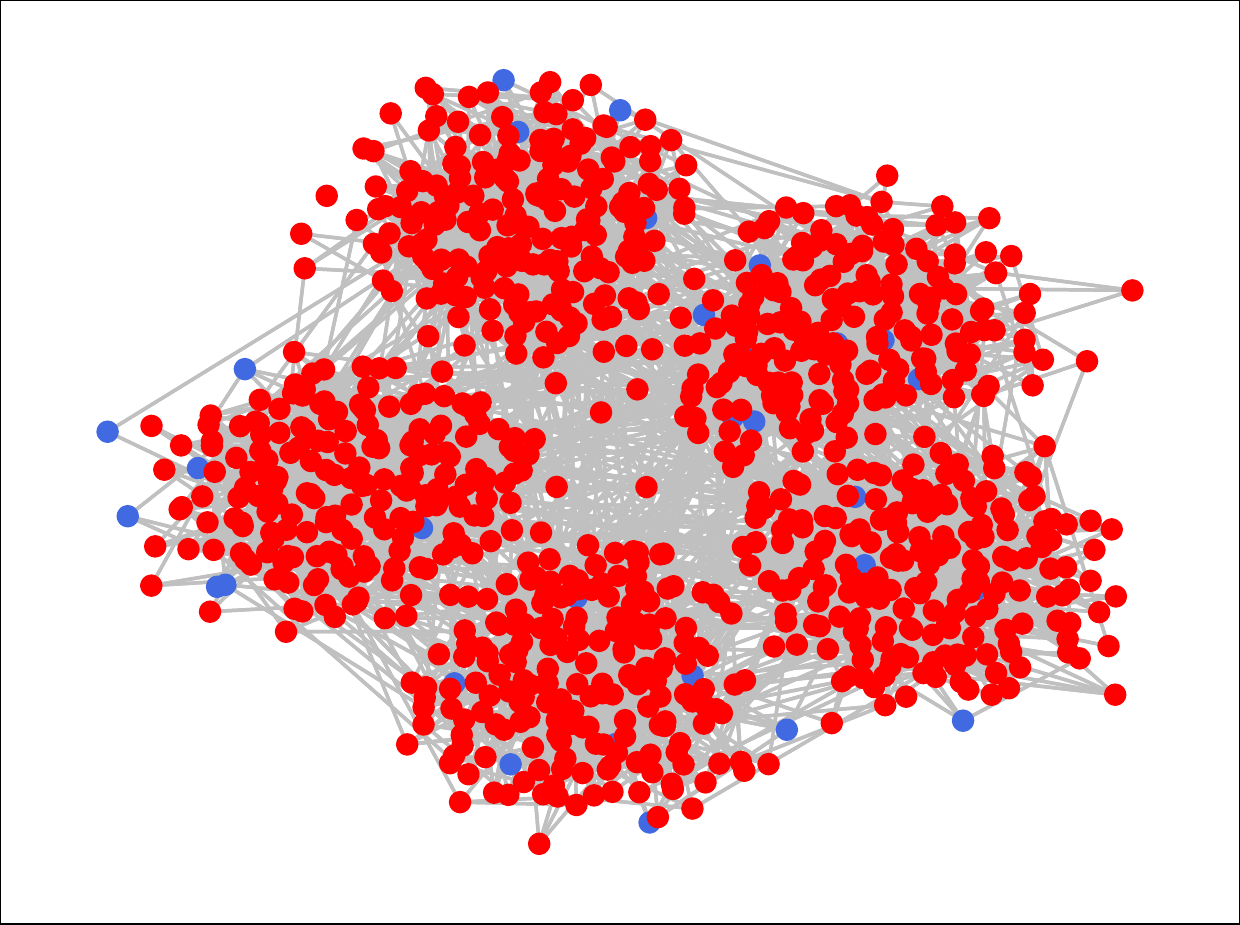}
		&   \includegraphics[width=\hsize,valign=m]{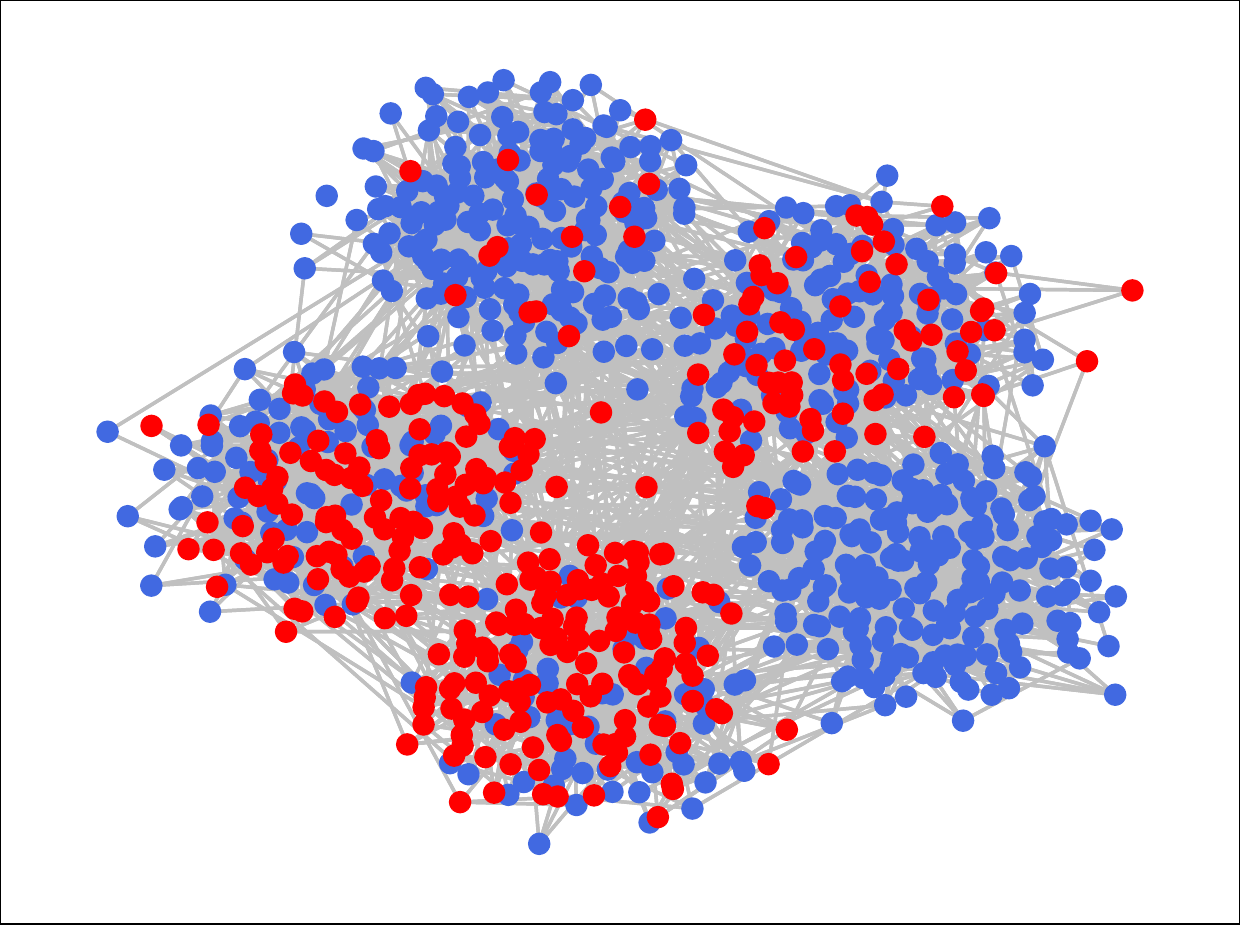}
		&   \includegraphics[width=\hsize,valign=m]{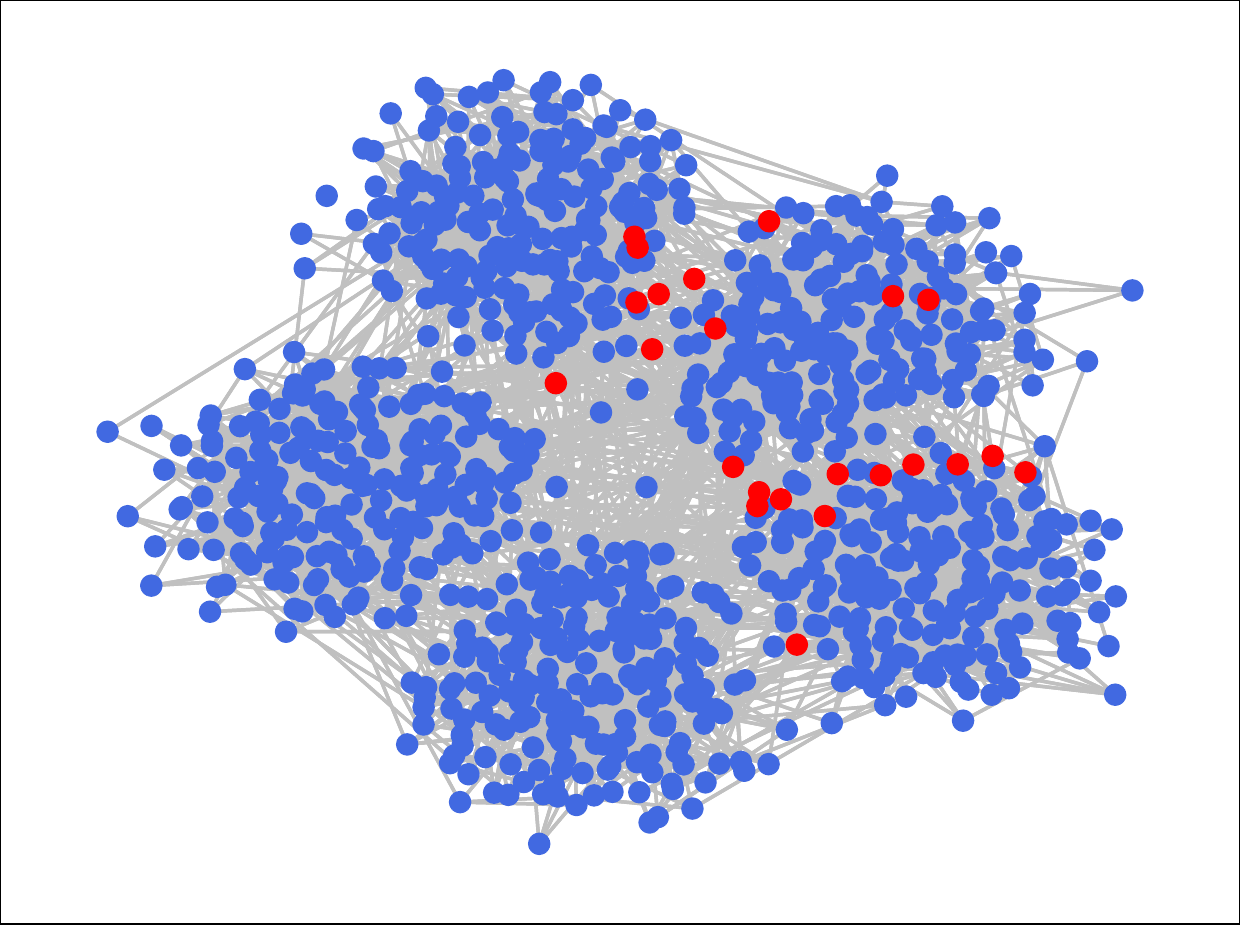}       
		&   \includegraphics[width=\hsize,valign=m]{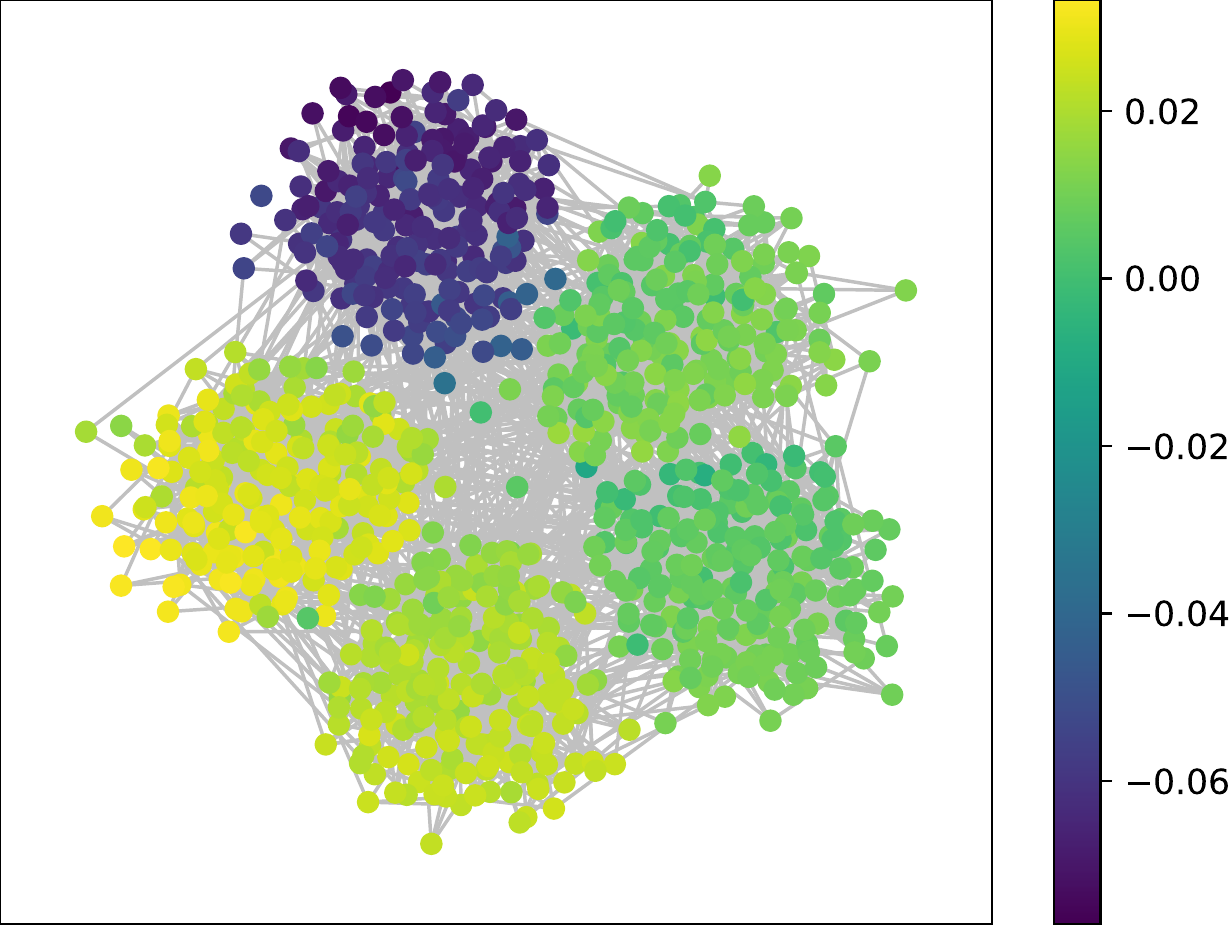}     \\  \addlinespace[0pt]
		\rothead{\centering\footnotesize Random \\Geometric} 
		&   \includegraphics[width=\hsize,valign=m]{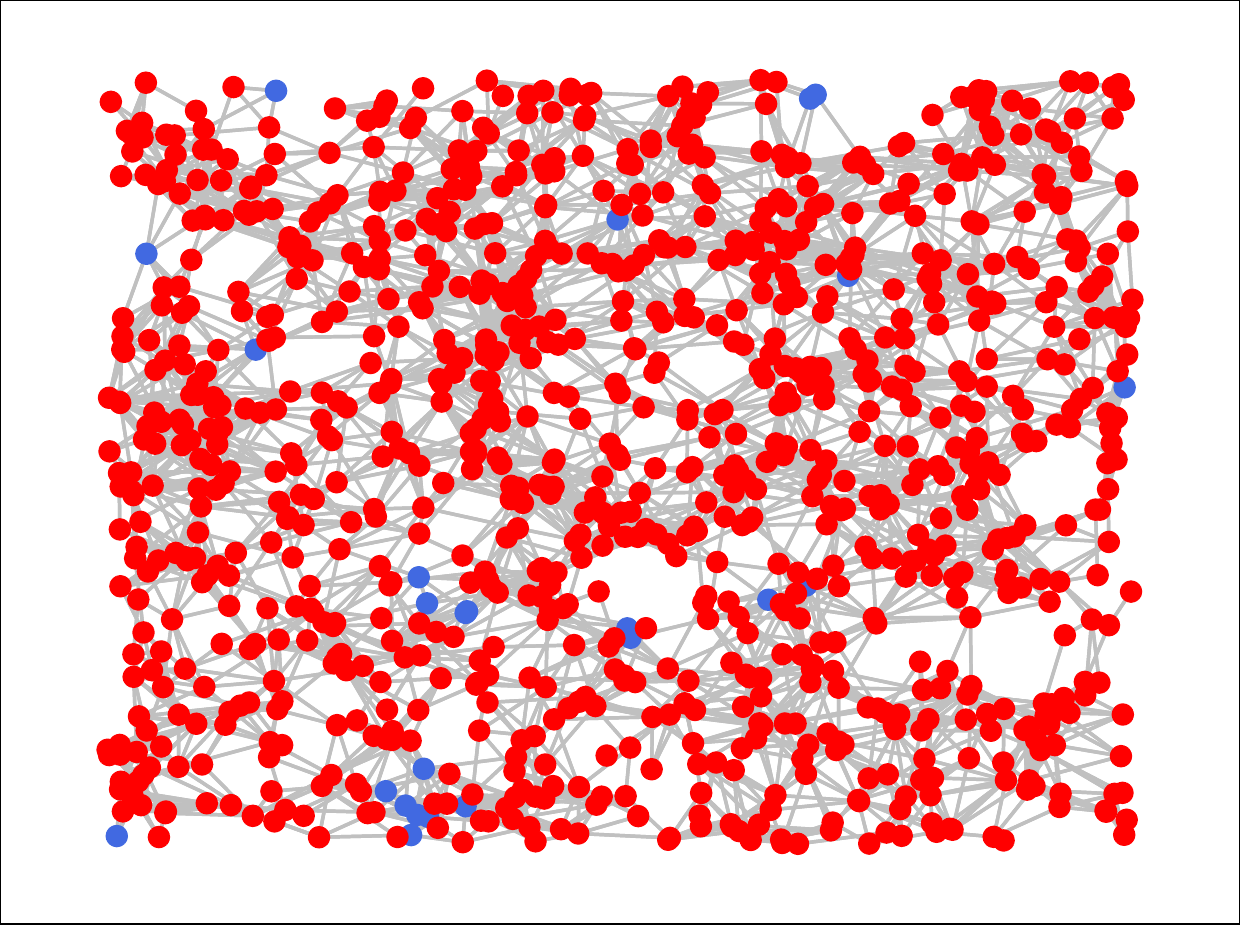}
		&   \includegraphics[width=\hsize,valign=m]{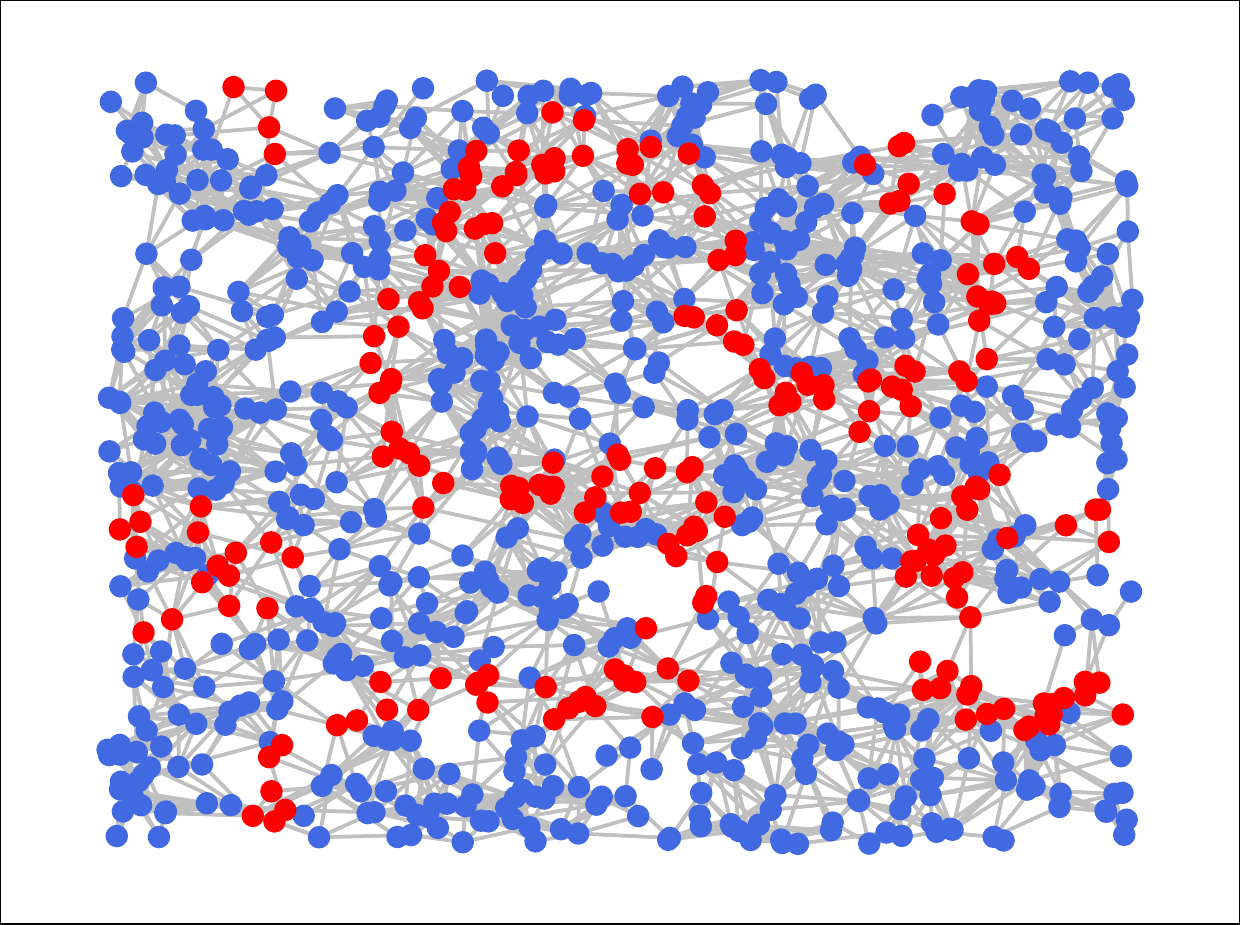}
		&   \includegraphics[width=\hsize,valign=m]{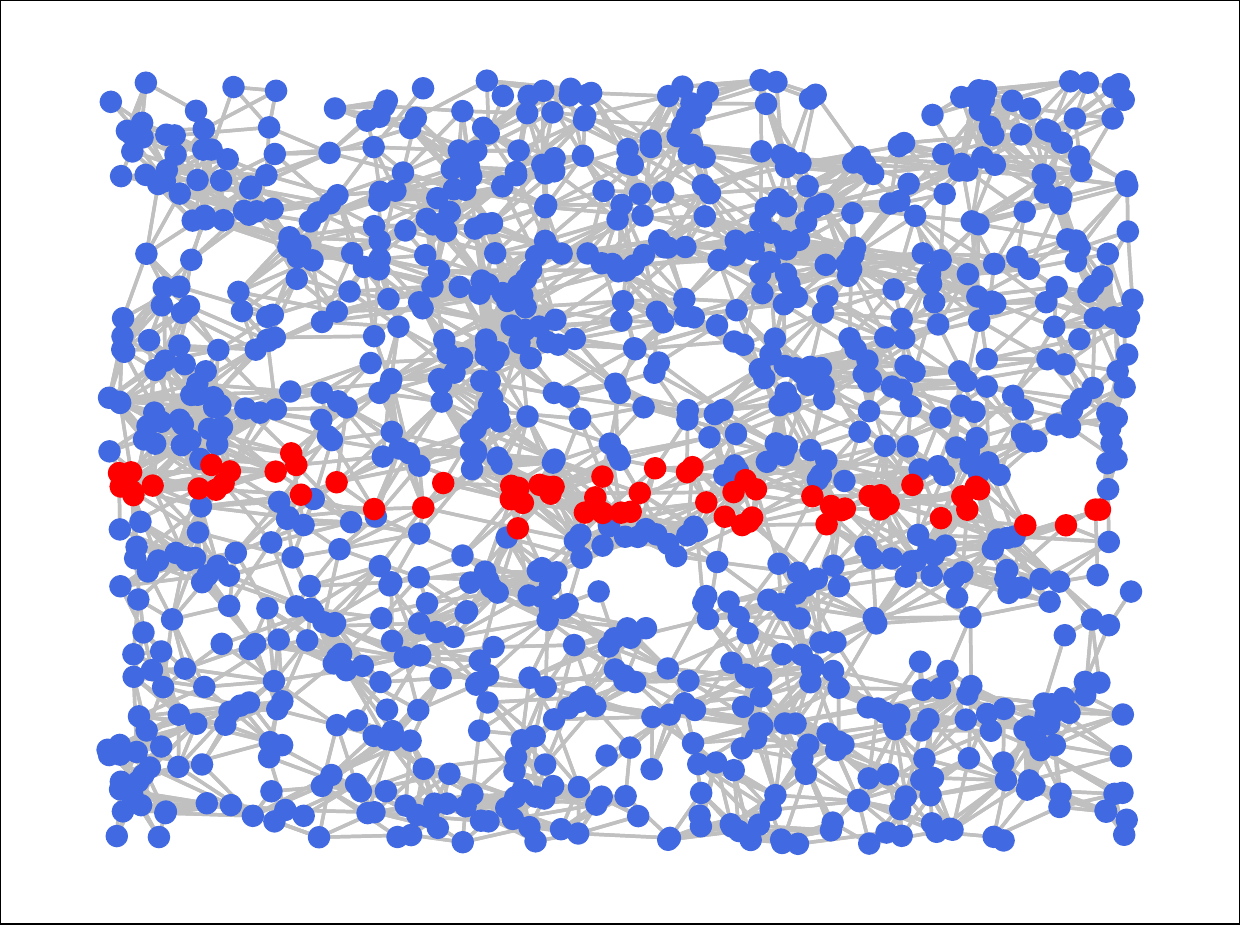}       
		&   \includegraphics[width=\hsize,valign=m]{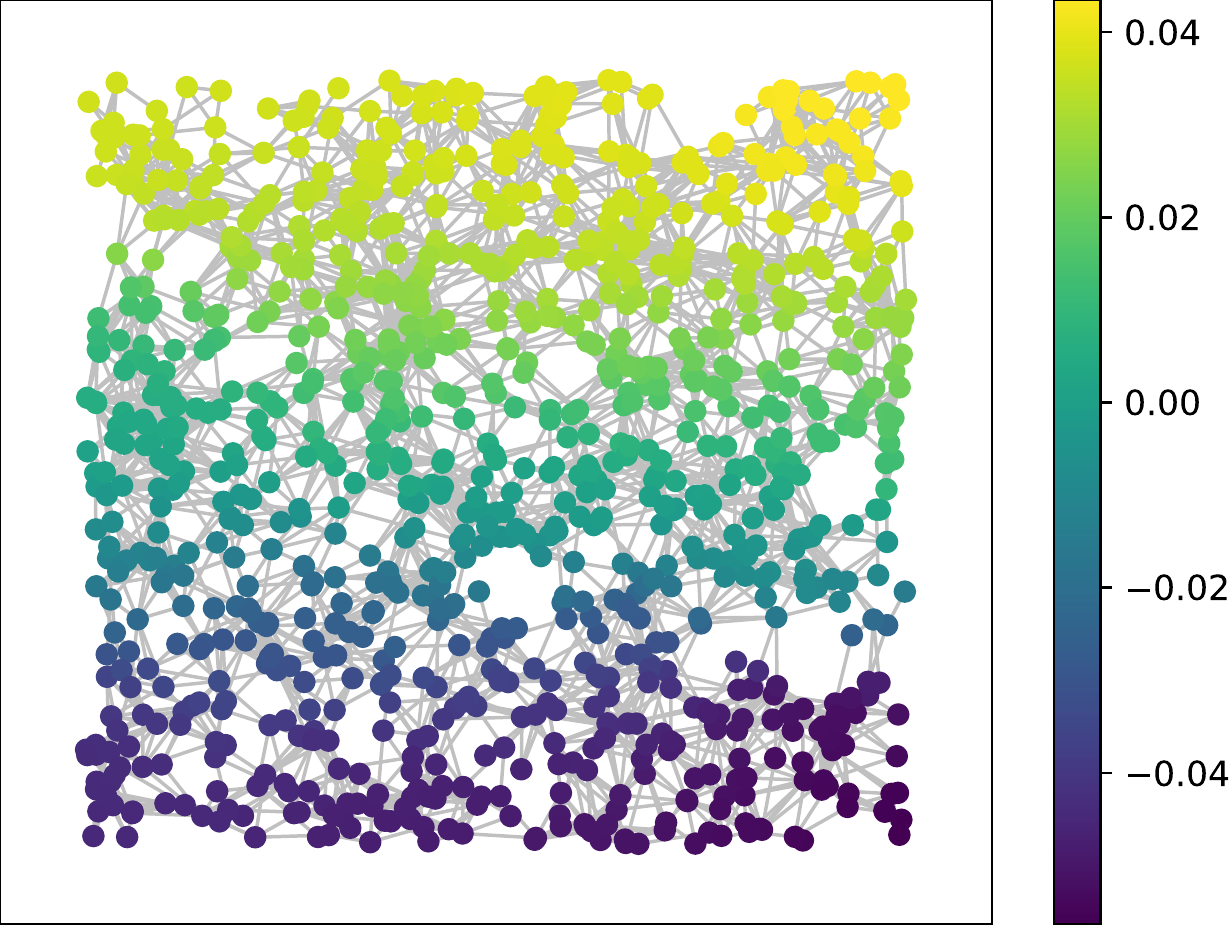}     \\  \addlinespace[0pt]
		\rothead{\centering\footnotesize Swarthmore \\Facebook} 
		&   \includegraphics[width=\hsize,valign=m]{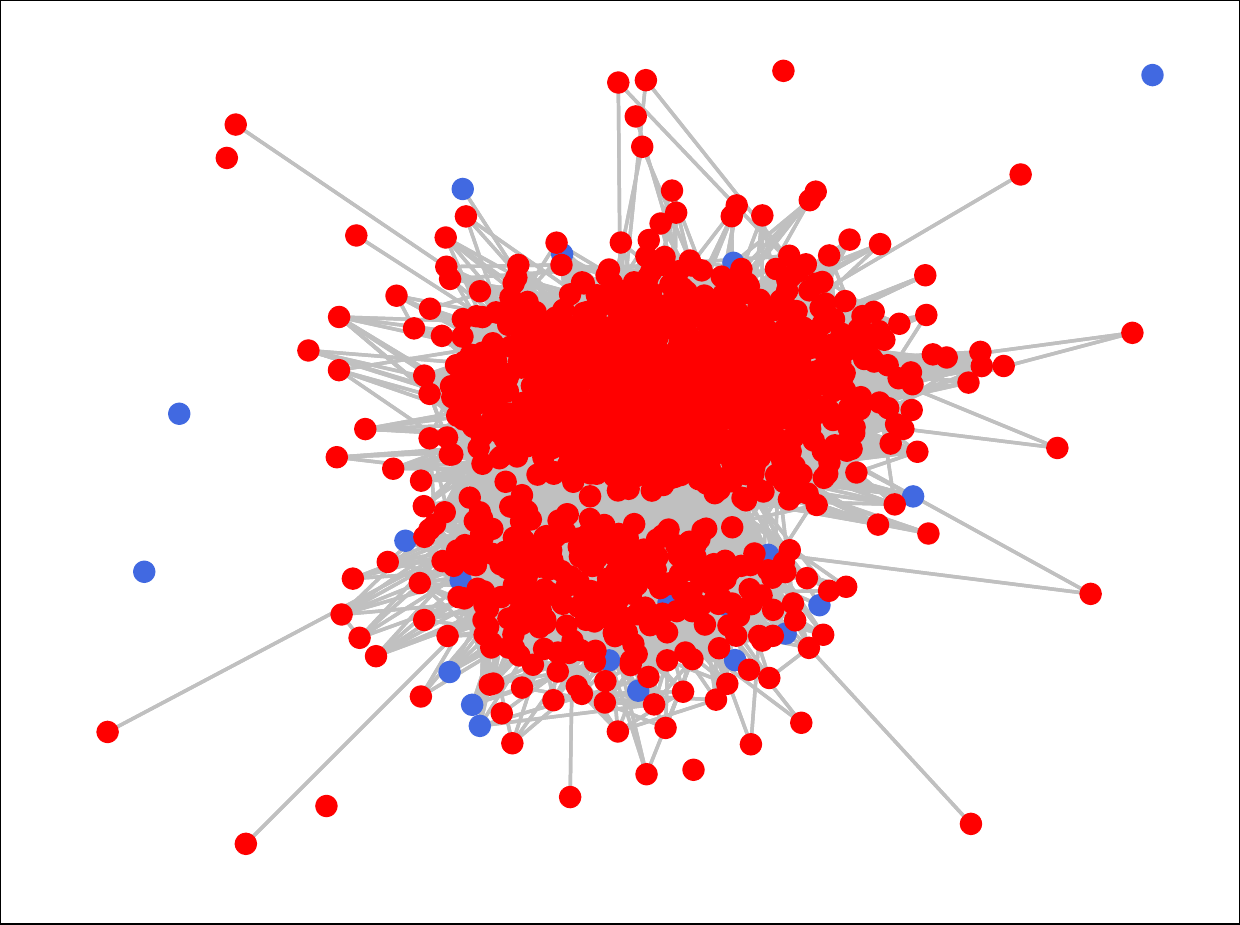}
		&   \includegraphics[width=\hsize,valign=m]{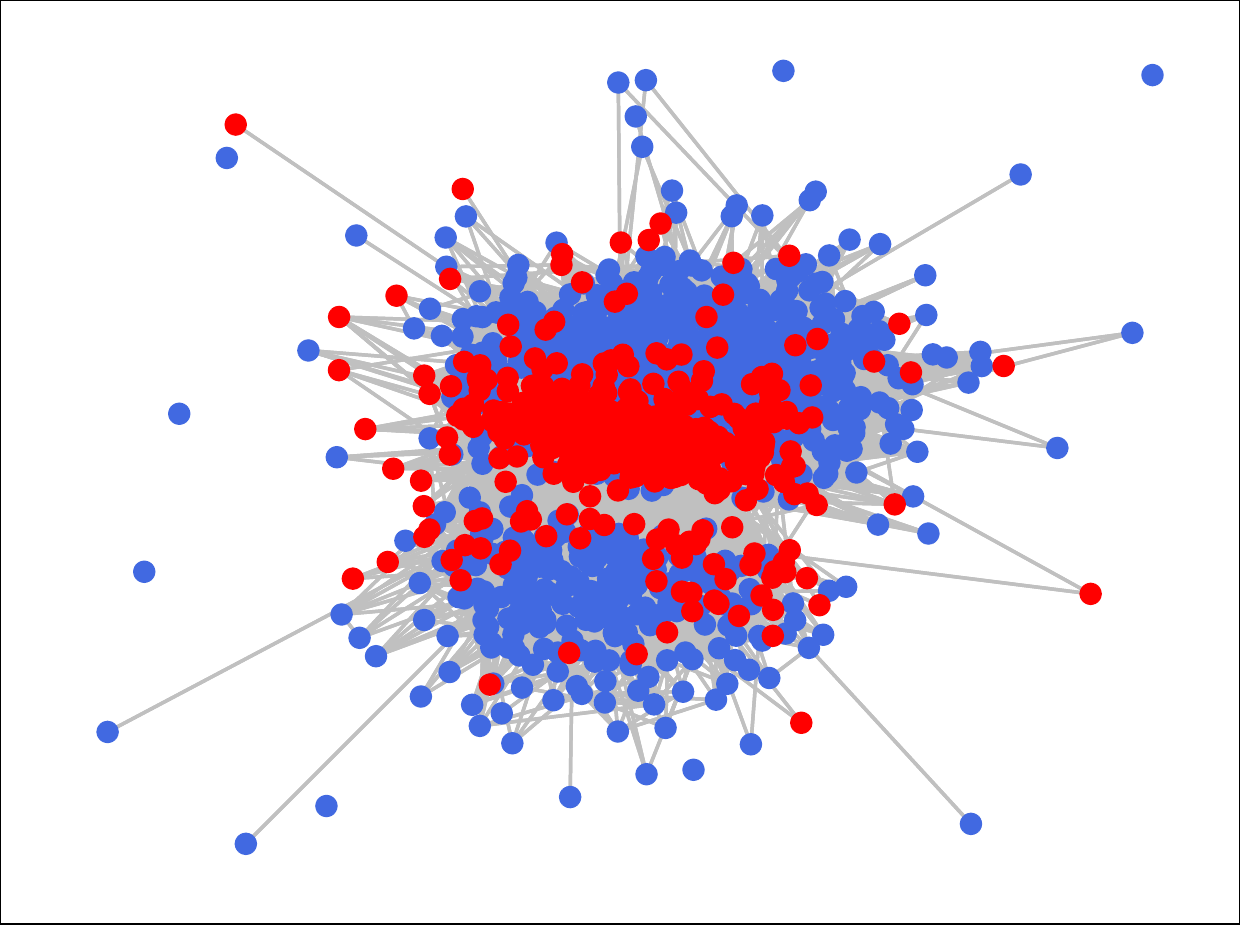}
		&   \includegraphics[width=\hsize,valign=m]{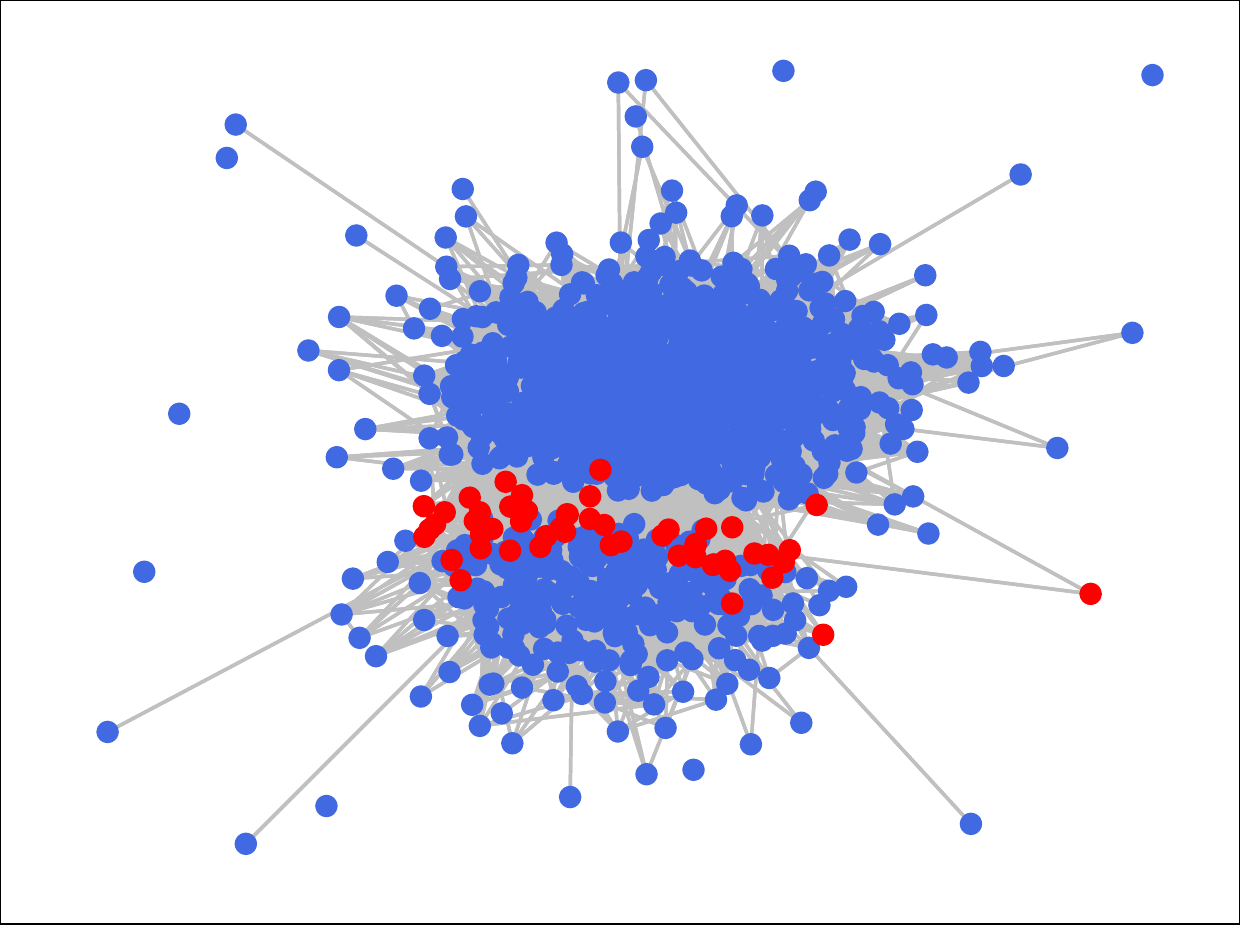}       
		&   \includegraphics[width=\hsize,valign=m]{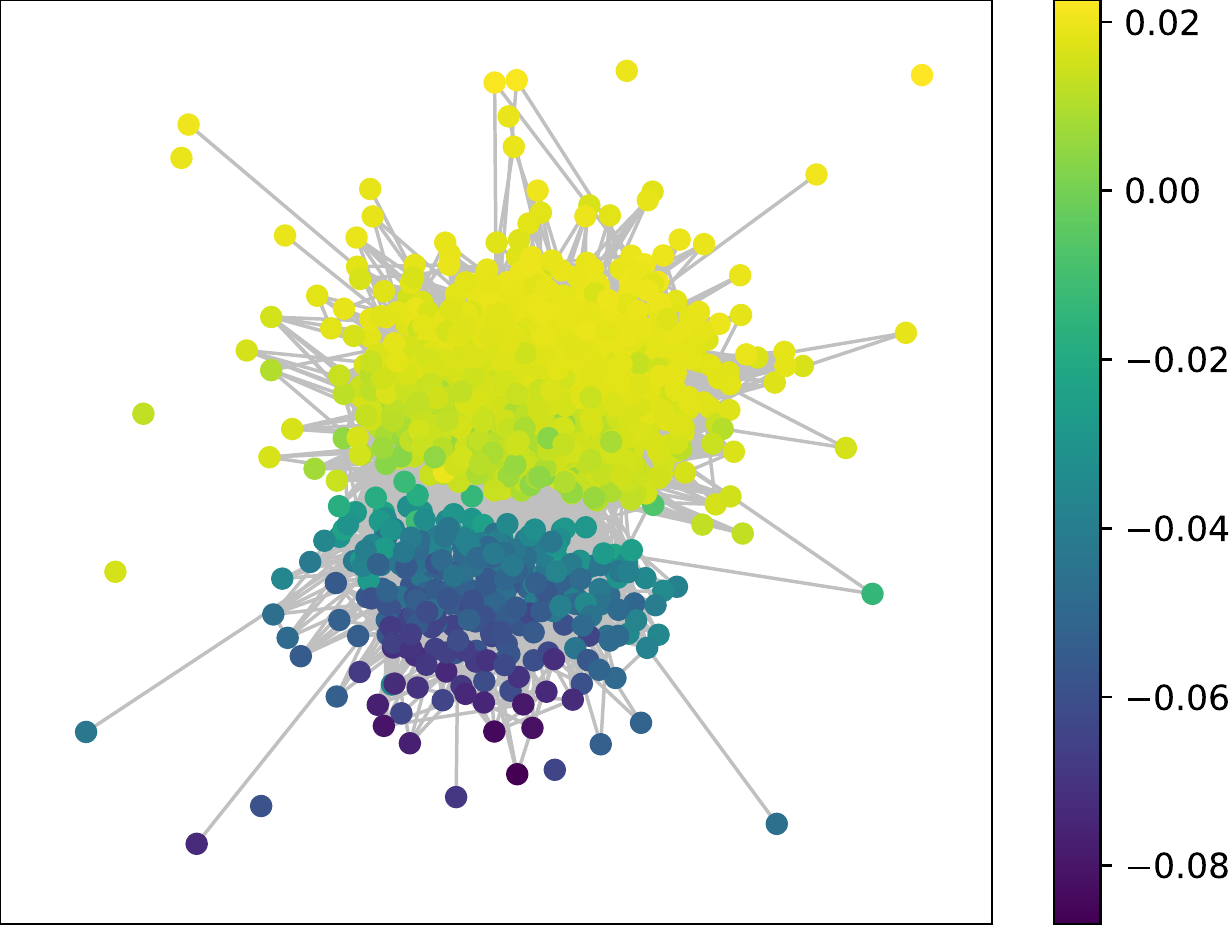}     \\  \addlinespace[0pt]
	\end{tabularx}
	\caption{A visualization of local agreement by number
		of DeGroot iterations for three social networks. Nodes are colored blue for individuals
		with at least 2/3 of their neighbors on the same side
		of the mean opinion, and red otherwise. The last column shows the normalized differences from the
		mean opinion at equilibrium 
		($\bar{\mathbf{s}}^*$ 
		from Theorem \ref{thm:second_eig_scaling}). In all cases, strong clusters of high local agreement emerge, which may lead to increased perceptions of opinion polarization.
	}
	\label{fig:local}
\end{figure}

%


To better understand the steep increase in this group-based polarization metric theoretically, we show that for an unweighted, regular graph $G$,  average local agreement has a simple linear algebraic form. Ultimately, the following claim will help us relate the measure to spectral properties of the underlying social graph $G$.

\begin{claim} \label{def:local_agreement}
	Let $G$ be an unweighted $d$-regular graph with no self-loops. Let $\mathbf{z}$ be a vector of opinions and let $\mathbf{s}= \sign(\mathbf{z} - \mean(\mathbf{z}) \cdot \vec{\mathbf{1}})$. Then, the average local agreement $\mathcal{L}(G,\mathbf{z})$ equals:
\begin{align*} 
	\mathcal{L}(G,\mathbf{z})  &= \frac{\mathbf{s}^T\mathbf{A}\mathbf{s}}{2nd} + \frac{1}{2}  &\text{where \,\,} \mathbf{s} &= \sign(\mathbf{z} - \mean(\mathbf{z}) \cdot \vec{\mathbf{1}}).
\end{align*}
\end{claim}
\begin{proof}
	For a node $i$, let $p_i= \sum_{j\in \mathcal{N}(i)} \mathbbm{1}[s_j = +1]$ denote the number of nodes in $\mathcal{N}(i)$ that are on the positive side of the mean and let $q_i = \sum_{j\in \mathcal{N}(i)} \mathbbm{1}[s_j = -1]$ denote the number of nodes on the negative side of the mean.  Let $a_i$ denote the number of nodes in $\mathcal{N}(i)$ that agree with node $i$ (i.e., are on the same side of the mean) and let $b_i$ denote the number of nodes that disagree. We can write: 
	\begin{align*}
		a_i &=  \begin{cases} 
			p_i & \text{ if } s_i = +1 \\
			q_i & \text{ if } s_i = -1 \\
		\end{cases} 
	&b_i =  \begin{cases} 
		p_i & \text{ if } s_i = -1 \\
		q_i & \text{ if } s_i = +1 \\
	\end{cases} 	
	\end{align*}
Observe that the $i^\text{th}$ entry of $\mathbf{A}\mathbf{s}$ equals $p_i - q_i$, and thus:
\begin{align*}
   \mathbf{s}^T\mathbf{A} \mathbf{s}
   = \sum_{i =1}^n s_i (p_i - q_i) = \sum_{i =1}^n a_i - b_i.
\end{align*}
Next note that $a_i + b_i = d$ and thus $nd = \sum_{i =1}^n a_i + b_i$. So we have $\mathbf{s}^T\mathbf{A}\mathbf{s} + nd = \sum_{i =1}^n 2a_i$.
Dividing by $2nd$ gives the result because $\mathcal{L}(G,\mathbf{z}) = \frac{1}{n} \sum_{i =1}^n \frac{a_i}{d}$.
\end{proof}

With Claim \ref{def:local_agreement} in place, we make the following observation:
\begin{observation}\label{obs:local_limit_2ndeig}
	For an unweighted graph $G$, we can approximate the equilibrium average local agreement $\lim_{t\rightarrow \infty} \mathcal{L}(G,\mathbf{z}^{(t)})$ by 
	\begin{align*}
		\lim_{t\rightarrow \infty} \mathcal{L}(G,\mathbf{z}^{(t)}) \approx \frac{\lambda_2}{2} + \frac{1}{2},
	\end{align*}
where $\lambda_2$ is the second eigenvalue of $G$'s normalized adjacency matrix.
\end{observation}
According to this claim, we expect to see high equilibrium local agreement -- i.e., \emph{increasing} polarization -- in any graph with a second eigenvalue close to $1$, which includes any graph with strong community, and thus most natural social networks. For example, the Facebook100 networks had an average second eigenvalue of $.871$. Observation \ref{obs:local_limit_2ndeig} predicts that this would lead to a mean equilibrium average local agreement of approximately $.936$, which is extremely close to the observed mean of $.948$.

To establish Observation \ref{obs:local_limit_2ndeig}, again assume that $G$ is $d$-regular with no self-loops. Note that for a regular graph, the second eigenvalue of $\mathbf{D}^{-1}\mathbf{A}$ is equal to that of $\mathbf{D}^{-1/2}\mathbf{A}\mathbf{D}^{-1/2}$. 
By Corollary \ref{corollary:invariant}, we have that 
$\lim_{t\rightarrow \infty} \mathcal{L}(G,\mathbf{z}^{(t)}) =\mathcal{L}(G,\mathbf{v}_2)$. And by Claim \ref{def:local_agreement}:
\begin{align*}
	\mathcal{L}(G,\mathbf{v}_2) &= \frac{\sign(\mathbf{v}_2^T) \mathbf{A} \sign(\mathbf{v}_2)}{2nd} + \frac{1}{2} \\
	&= \frac{\sign(\mathbf{v}_2^T) \mathbf{D}^{-1/2}\mathbf{A}\mathbf{D}^{-1/2} \sign(\mathbf{v}_2)}{2n} + \frac{1}{2}
\end{align*}
Observation \ref{obs:local_limit_2ndeig} then immediately follows by noticing that
\begin{align*}
\sign(\mathbf{v}_2^T) \mathbf{D}^{-1/2}\mathbf{A}\mathbf{D}^{-1/2} \sign(\mathbf{v}_2) \approx n \mathbf{v}_2^T\mathbf{D}^{-1/2}\mathbf{A}\mathbf{D}^{-1/2} \mathbf{v}_2 = n\lambda_2.
\end{align*}
The approximation is exact if all entries in the unit vector $\mathbf{v}_2$ have magnitude $1/\sqrt{n}$. It tends to hold a close approximation in other networks, which can be formalized via well-known connections between balanced cut problems and the second eigenvector \cite{McSherry:2001,Spielman:2019}.

\begin{figure}
	\centering
	\includegraphics[scale=.5]{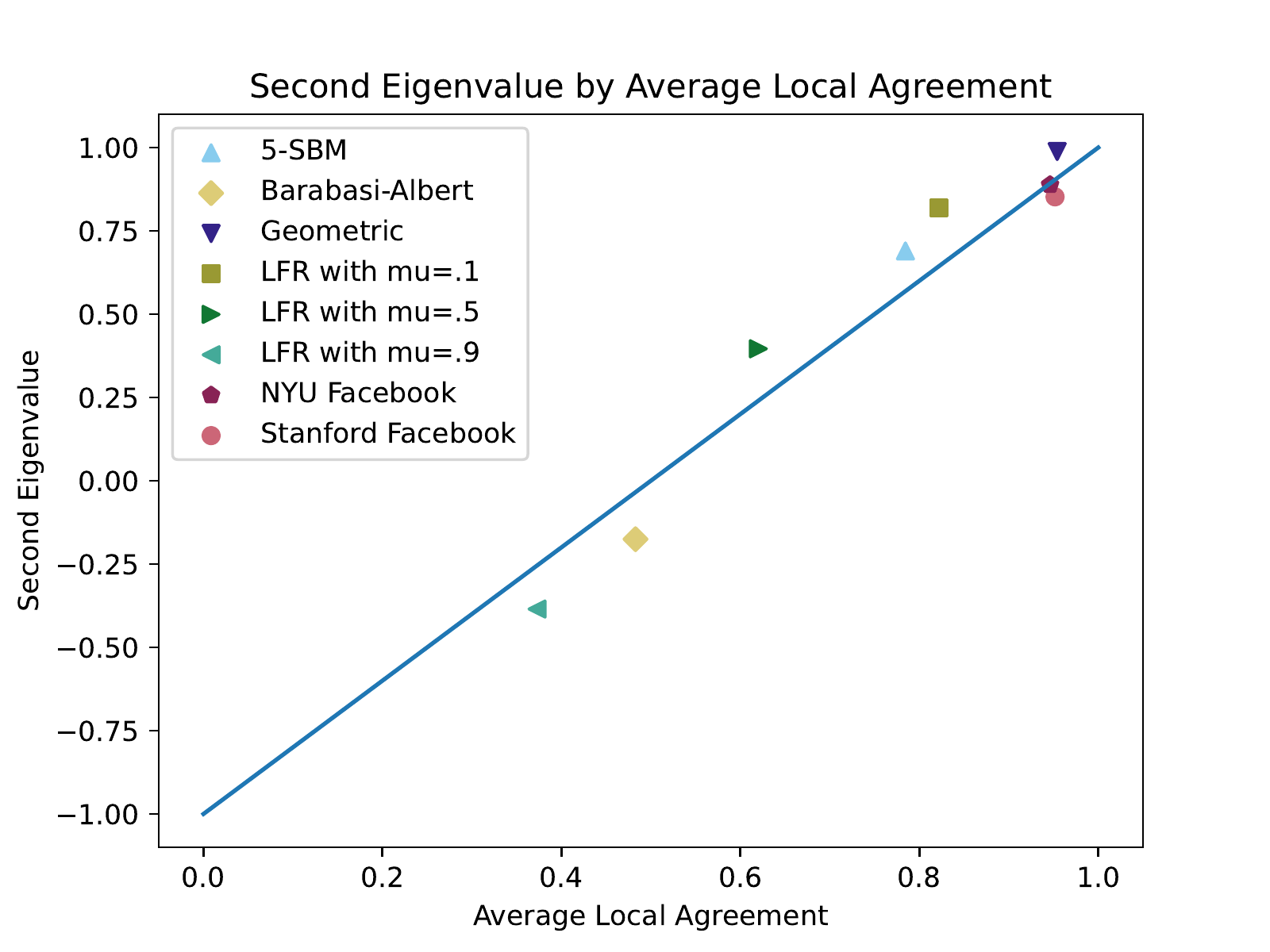}
	\vspace{-1em}
	\caption{Average local agreement at equilibrium
		plotted against the second normalized  adjacency matrix eigenvalue for several random graphs generated with the NetworkX package \cite{HagbergSchultSwart:2008}, and for the NYU and Stanford Facebook graphs. The values closely align with the linear relationship predicted by Observation 
		\ref{obs:local_limit_2ndeig} (plotted as a solid line).
	}
	\label{fig:ratio}
\end{figure}

In Figure \ref{fig:ratio}, we empirically confirm
the relationship described
in Observation \ref{obs:local_limit_2ndeig} by examining a variety of random graphs with widely varying second eigenvalue. 
We find that the correlation between average local agreement
and second eigenvalue in the Facebook100 data set is statistically
significant  ($p=5e^{-5}$)
with a Pearson correlation of $r=.392$.


Finally, we comment on the rate at which average local agreement converges to its equilibrium value. Since this rate depends on how quickly the normalized difference vector converges to $\bar{s}^*$ under the DeGroot model, we expect it to scale linearly with the inverse of the \emph{second} eigenvalue gap, $\frac{|\lambda_2|-|\lambda_{3}|}{|\lambda_2|}$. We confirm this relationship on the Facebook100 data set, where we see a statistically significant ($p=7e^{-6}$) correlation
between inverse second eigengap and average number of
iterations until convergence to the final average local agreement when starting with a random opinion vector. The Pearson correlation coefficient of the relationship is $r=.451$.

In contrast, the rate at which the opinion vector converges to $\mathbf{z}^*$ depends inversely on the \emph{first} eigengap $\frac{|\lambda_1|-|\lambda_{2}|}{|\lambda_1|}$. As such, when the second eigengap is large compared to the first, we expect local agreement to increase more quickly than opinion variance decreases, which might contribute to perceptions of growing polarization.




\section{Conclusion}
In this work we established that  natural group-based polarization measures display interesting dynamics under the standard DeGroot opinion formation model. Unlike heavily studied variance-based measures, we showed both empirically and theoretically that group-based measures can increase over time, and often do increase quite significantly in natural social networks.

We leave a number of questions for future research. As discussed, recent work on mathematical models of opinion dynamics has sought to understand the impact of outside actors (who can modify the graph $G$ is some way) on individual opinions and polarization \cite{GaitondeKleinbergTardos:2021,AbebeChanKleinberg:2021}. There is little work on how such modifications impact group-based polarization, and if they can accelerate its emergence.

Another challenging question it to determine the ``right'' group-based measure of polarization for use in opinion dynamics studies. Our work establishes that group-based measures in general are more in line with real-world polarization than variance-based measures, but more in-depth studies of e.g. political opinion surveys would be needed to determine exactly \emph{which} group-based measures best align with perceived polarization. Currently, there is some experimental evidence for the value of ideological alignment as a meaningful polarization metric \cite{levendusky2009partisan,FiorinaAbramsPope:2005}, but statistical measures of bimodality and ``local'' metrics have received less attention.

\bibliography{references}
\bibliographystyle{ACM-Reference-Format}

\end{document}